\documentclass[draftclsnofoot, onecolumn, comsoc, 12pt]{IEEEtran} 
\IEEEoverridecommandlockouts
\usepackage{amsthm}
\usepackage{amsmath}
\usepackage{algorithmic}

\usepackage[T1]{fontenc}

\usepackage{graphicx}
\usepackage[noadjust]{cite}
\usepackage{mcite}
\usepackage{amsfonts,helvet}
\usepackage{fancyhdr}
\usepackage{threeparttable}
\usepackage{epsf,epsfig}
\usepackage{amsthm}
\usepackage{amsmath}
\usepackage{siunitx}
\usepackage{amssymb}
\usepackage{dsfont}
\usepackage{color}
\usepackage{subcaption}
\usepackage{enumerate}
\usepackage{hyperref}
\usepackage{cancel}
\usepackage{bbm}
\usepackage{dsfont}
\usepackage[subnum]{cases}
\usepackage{adjustbox}
\usepackage[linesnumbered,noend,ruled]{algorithm2e}
\usepackage{multicol}
\usepackage[english]{babel}
\usepackage{enumitem}

\newtheorem{lemma}{Lemma}

\newtheorem{proposition}{Proposition}
\newtheorem{remark}{Remark}

\newmuskip\pFqmuskip

\def\bff{{\bf f}}
\def\bg{{\bf g}}
\def\bh{{\bf h}}

\def\bs{{\bf s}}

\def\bx{{\bf x}}

\def\bA{{\bf A}}
\def\bB{{\bf B}}
\def\bC{{\bf C}}
\def\bD{{\bf D}}

\def\bF{{\bf F}}

\def\bR{{\bf R}}

\def\cL{\mbox{$\mathcal{L}$}}

\def\cU{\mbox{$\mathcal{U}$}}

\def\bbC{\mbox{$\mathbb{C}$}}

\def\bbE{\mbox{$\mathbb{E}$}}

\def\bbN{\mbox{$\mathbb{N}$}}

\def\bbR{\mbox{$\mathbb{R}$}}

\usepackage{eucal}

\begin{document}

\title{Joint Optimization for Secure and Reliable Communications in Finite Blocklength Regime
}

\author{Mintaek Oh, Jeonghun Park, and Jinseok Choi
\thanks{
M. Oh and J. Choi are with the Department of Electrical Engineering, Ulsan National Institute of Science and Technology (UNIST), Ulsan, 44919, South Korea (e-mail: {\texttt{\{ohmin, jinseokchoi\}@unist.ac.kr}}). 

J. Park is with the School of Electronics Engineering, College of IT Engineering, Kyungpook National University, Daegu, 41566, South Korea (e-mail: {\texttt{jeonghun.park@knu.ac.kr}}). 
}
}

\maketitle \setcounter{page}{1} 
\begin{abstract} 
To realize ultra-reliable low latency communications with high spectral efficiency and security, we investigate a joint optimization problem for downlink communications with multiple users and eavesdroppers in the finite blocklength (FBL) regime.
We formulate a multi-objective optimization problem to maximize a sum secrecy rate by developing a secure precoder and to minimize a maximum error probability and information leakage rate. 
The main challenges arise from the complicated multi-objective problem, non-tractable back-off factors from the FBL assumption, non-convexity and non-smoothness of the secrecy rate, and the intertwined optimization variables.
To address these challenges, we adopt an alternating optimization approach by decomposing the problem into two phases:  secure precoding design, and maximum error probability and information leakage rate  minimization.
In the first phase, we obtain a lower bound of the secrecy rate and 
 derive a first-order Karush–Kuhn–Tucker (KKT) condition to identify local optimal solutions with respect to the precoders.
Interpreting the condition as a generalized eigenvalue problem, we solve the problem by using a power iteration-based method.
In the second phase, we adopt a weighted-sum approach and derive KKT conditions in terms of the error probabilities and  leakage rates  for given precoders.
Simulations validate the proposed algorithm.
\end{abstract}

\begin{IEEEkeywords}
    Physical layer security, finite blocklength, secure precoding, error probability and information leakage rate minimization, alternating optimization.
\end{IEEEkeywords}

\section{Introduction}








Ultra-reliable low latency communications (URLLC) have inevitably become one of the primary usage scenarios in realizing  5G and 6G communications~\cite{URLLC:TCOM:2019}.
For instance, URLLC is closely related to mission-critical internet-of-things (IoT) and vehicle-to-everything (V2X) applications oriented toward high reliability
and low latency~\cite{V2X:MAG:2015, URLLC:TCOM:2020}.
To support such delay-sensitive communications, a short-packet transmission is often taken into account \cite{SPC:Proc:2016, SPC:TCOM:2016, SPC:TVT:2018, SPC:TWC:2019}.
Unlike conventional communications that mainly consider the infinite blocklength coding regime, the finite blocklengh (FBL) based communications is more suitable to leverage the benefit of short data packet transmission~\cite{FBL:CL:2018, FBL:CL:2017}.
The communication in the FBL regime, however, is limited by a back-off factor determined by non-negligible decoding error probability and blocklength \cite{FBL:TIT:2010}.
Accordingly, the existing transmission strategies with infinite blocklength becomes far less optimal in the FBL regime and thus, delay-constrained applications need to carefully consider the reliability issue as well as the latency; optimizing the per-user error probability plays a key role in achieving the high throughput with low latency in the FBL regime \cite{FBL:WCL:2017}.

In addition,
information security has become a  critical issue for future wireless networks to protect confidential information from eavesdroppers.
Since the wireless communication systems are highly likely to be vulnerable to eavesdropping due to the broadcast nature of wireless communications \cite{PLS:NET:2015}, the importance of security has been gaining more attention~\cite{zhang2014iot, hassan2019current}.
Considering the complexity issues, physical layer security (PLS) has been considered as a promising solution for secure communications \cite{PLS:NET:2015, PLS:JSAC:2018}.
In \cite{PLS:JSAC:2018}, a key mechanism of PLS is to apply a secure channel coding technique to confidential information so that only the legitimate users are capable of decoding the information.
In \cite{PLS:BELL:1975, PLS:TIT:1978}, the maximal secret communication rate called the secrecy rate has characterized at which information can be transmitted securely and reliably over a wiretap channel.
In spite of the broad investigation on the PLS methods, it is inappropriate to directly apply the existing solutions when communication systems operate in the FBL regime due to the additional back-off factor from the FBL penalty in PLS~\cite{PLS_FBL:ISIT:2016, PLS_FBL:TIT:2019, feng2021secure}.
In this regard, we investigate PLS for downlink FBL-based communications.




\subsection{Prior Work}
PLS has been widely studied in wireless communication systems.
For the case of multiple users and a single eavesdropper in multiple-input single-output (MISO) wireless networks, suboptimal precoding solutions were proposed in previous works.
A secure precoding that maximizes a sum secrecy spectral efficiency (SE) with suboptimal scheduling was proposed in \cite{PLS:SPL:2013}.
In \cite{PLS:TSP:2018}, a secrecy rate optimization was proposed for maximizing the MISO network’s energy efficiency in terms of users’ quality-of-service (QoS).
For a multiple-input multiple-output (MIMO) system, the secrecy SE with multiple eavesdroppers equipped with multiple antennas was investigated in \cite{PLS:TIT:2010}.
Furthermore, in \cite{PLS:TWC:2015}, a secure precoding algorithm was proposed in MIMO broadcasting channels for a simultaneous wireless information and power transfer system.
For a multi-user multi-eavesdropper network, a secure transmission strategy was proposed in \cite{PLS:TIFS:2016}. 
Specifically, in \cite{PLS:TIFS:2016}, an artificial noise (AN)-aided secure precoding was considered to maximize the secrecy throughput by addressing the leakage signal from non-colluding eavesdroppers.
Moreover, there exists prior work incorporating a scenario of colluding eavesdroppers~\cite{PLS:TVT:2021}.

Based on the analytical finding of the maximal channel coding rate at the FBL \cite{FBL:TIT:2010}, there have been several works that analyzed accurate information-theoretic approximations of the achievable rates in the FBL regime \cite{FBL:ITW:2012, FBL:TIT:2014, FBL:TVT:2015}.
Leveraging the results in the FBL regime, in \cite{FBL:TVT:2019, FBL:TCOM:2020}, resource allocation   was investigated to maximize the sum rate and to assist  the short-packet communication subject to QoS constraints for URLLC users. 
In addition, the decoding error probability  also needs to be considered in advanced wireless communication techniques incorporating reliable and real-time services.
In this regard, the joint optimization of SE and error probability in the FBL regime is needed to cope with the stringent requirements of URLLC.
The joint optimization problem in \cite{choi:iotj:21} was solved by using an alternation optimization framework between the sum SE maximization and error probability minimization.

Recent studies have also analyzed PLS in the FBL regime \cite{PLS_FBL:ISIT:2016, PLS_FBL:TIT:2019}.
In \cite{PLS_FBL:TIT:2019}, the maximal secret communication rate over a wiretap channel in the FBL regime, i.e., the secrecy rate, was derived.
According to the fundamental results, it was observed that there exist tradeoffs among delay, reliability, and security.
In this regard, the wireless communication techniques need to be further investigated for improving the secrecy rate  in the FBL regime \cite{PLS_FBL:PROC:2021}.
To find the optimal tradeoff, state-of-the-art precoding methods were proposed to improve the secrecy rate under the reliability and security constraints.
For the case of a single user and eavesdropper, analytical frameworks of the  FBL-based  secure communications were investigated in \cite{FBL:TWC:2019}.
Specifically, in \cite{FBL:TWC:2019}, the performance of FBL-based communications in the presence of an eavesdropper was investigated, and the optimal blocklength in maximizing secrecy throughput was also analyzed by using the AN-aided maximum ratio transmission (MRT) precoding.
For full-duplex MIMO short-packet systems,  \cite{PLS_FBL:WCL:2021} proposed the secure precoding method based on zero-forcing (ZF) to maximize the sum secrecy rate.
From \cite{PLS_FBL:TWC:2021}, an outage probability considering both reliability and security was proposed according to the characteristics of  FBL-based secure communications.
In our previous work \cite{oh2022secure}, we proposed an efficient secure precoding algorithm for downlink communications with multiple users and a single eavesdropper.

Although the prior works investigated the secure communications in the FBL regime,  the single-eavesdropper scenario is mainly considered.
In addition, the prior work  provided theoretical analysis and made an effort of maximizing the achievable secrecy rate with the conventional linear precoders such as MRT and ZF.
In this regard, thorough optimization of transmission for  FBL-based secure communications under the reliability and security constraints for the general multi-user multi-eavesdropper  MIMO network is still missing.
It is, however, challenging to solve the optimization problem for such sophisticated networks.
In addition, we note that solving the sum secrecy rate optimization problem is already  non-convex and difficult to solve  \cite{PLS:TWC:2016}, and  the back-off factors entangled with the secrecy rate make the problem more challenging \cite{PLS_FBL:TIT:2019}.
Besides, it is infeasible to find  direct solutions with the multi-objectives \cite{choi:iotj:21}.
Overcoming the aforementioned challenges, we put forth a novel secure precoding method that brings the optimal tradeoff among rate, reliability, and security in the FBL regime.

\subsection{Contributions}
In this paper, we consider the  FBL-based downlink secure communication systems where an access point (AP) with multiple antennas serves multiple users while multiple  eavesdroppers attempt to wiretap the user messages.
Our contributions are summarized as follows:
\begin{itemize}
    
    \item 
    Based on the characteristics of FBL-based communications, we adopt the secrecy rate as our key performance metric.
    Using the secrecy rate, we formulate a sum secrecy rate maximization problem to jointly optimize $(i)$ a precoding matrix, and $(ii)$ error probability and information leakage rate.
    However, there exist several challenges in solving the formulated optimization problem.
    First, the problem is inherently non-convex and thus, finding a global optimal solution is infeasible.
    Second, the formulated problem is a multi-objective optimization, and also each user has error probability and information leakage constraints determined by the system reliability requirements.
    Third, the secrecy rate is not tractable due to the back-off factors which are functions of optimization variables intertwined with each other.
    Finally, the sum secrecy rate is non-smooth since the secrecy rate under the presence of multiple eavesdroppers is determined by the maximum wiretap channel rate.

    \item 
    To resolve aforementioned difficulties, we first decompose the problem into two phases adopting alternating optimization: the secrecy rate maximization precoding design and the maximum error probability and information leakage rate minimization.
    In the first phase, for given error probability and  leakage rate, we apply a smooth approximation to a non-smooth objective function.
    Then, we make the problem a more tractable form by deriving a lower bound of the secrecy rate.
    Due to non-convexity of the problem, we derive a first-order KKT condition which can be interpreted as a generalized eigenvalue problem and hence, the best local optimal precoder can be found by finding its principal eigenvector.
    It can be solved by a power iteration-based precoding method.
    In the second phase, we reformulate the multi-objective optimization problem into a single-objective optimization problem by using a weighted-sum approach for a given precoder.
    Since the problem is still not tractable due to the maximum function in the secrecy rate, we further consider the lower bound of the objective function and optimize the maximum error probability and information leakage rate  by solving KKT conditions.

    \item
    We note that obtaining the instantaneous channel state information at a transmitter (CSIT) of eavesdroppers is challenging.
    Hence, we consider the case of the partial CSIT of wiretap channels which is regarded as a more practical scenario.
    Assuming that only the wiretap channel covariance is available at the AP, we reformulate the problem by deriving an approximated lower bound of the secrecy rate which can be handled by using the channel covariance.
    Then, we follow similar steps as for the perfect CSIT case and develop a joint precoder, error probability and information leakage rate optimization algorithm under the partial CSIT of wiretap channels.

    \item
    Via simulations, we verify the rate, reliability, and security performances of the proposed joint optimization methods.
    In particular, we demonstrate that the proposed algorithm achieves the highest secrecy rate  while keeping the lowest maximum error probability and information leakage rate in various scenarios compared to baseline methods.
    Therefore, it is concluded that the proposed algorithms are beneficial to the  FBL-based secure communications under the strict requirement of the reliability and security for future wireless applications.
\end{itemize}

{\it Notation}: $\bf{A}$ is a matrix and $\bf{a}$ is a column vector. 
The superscripts $(\cdot)^{\sf T}$, $(\cdot)^{\sf H}$, and $(\cdot)^{-1}$ denote the transpose, Hermitian, and matrix inversion, respectively.
The blackboard bold symbols $\bbC$, $\bbR_{+}$, and $\bbN_{+}$ denote the complex, nonnegative real, and nonnegative integer domains, respectively.
${\bf{I}}_N$ is the identity matrix with size $N \times N$.
Assuming that ${\bf{A}}_1, \dots, {\bf{A}}_N \in \mathbb{C}^{K \times K}$, ${\bf{A}} = {\rm blkdiag}\left({\bf{A}}_1,\dots, {\bf{A}}_N \right) \in \bbC^{KN\times KN}$ is a block diagonal matrix. $\|\bf A\|$ represents L2 norm. 
We use ${\rm{tr}}(\cdot)$ for trace operator, ${\rm{vec}}(\cdot)$ for vectorization, and
${\cU}(a,b)$ is a uniform distribution with two boundaries $a$ and $b$.
We also follow MATLAB style notation.


\section{System Model} \label{sec:sys_model}


We consider a downlink network in which the AP equipped with $N$ antennas serves $K$ single-antenna users. 
The network includes $M$ single-antenna eavesdroppers, and they attempt to overhear legitimate user messages from the AP.
In addition, we assume the FBL channel coding, i.e., the coding length is $L \ll \infty$.
We denote an user set and an eavesdropper set as $\CMcal{K}=\left\{1,\cdots,K\right\}$ and $\CMcal{M}=\left\{1,\cdots,M\right\}$, respectively. 
The data symbol for user $k$,  $s_k$, is drawn from a Gaussian distribution with a zero mean and variance of $\bbE [|s_k|^2] = P$, $\forall k \in \CMcal{K}$.

The AP broadcasts the data symbols $s_k, \forall k \in \CMcal{K}$ to each legitimate user through a linear precoder $\bF = [{\bf f}_1 , \ldots , {\bf f}_K] \in \bbC^{N \times K}$,
where ${\bf f}_{k} \in \bbC^{N}$ indicates a precoding vector for $s_k$, $\forall k \in \CMcal{K}$.
A transmitted signal vector ${\bx}\in \bbC^{N}$ is given by
\begin{align}
    {\bf{x}} = \sum_{k=1}^{K}{\bf{f}}_k s_k = {\bF\bs},
\end{align}
where $\bs = [s_1, \dots, s_K]^{\sf T} \in \bbC^{K}$.
After transmission, a received signal of  user $k$ is given by
\begin{align}
    y_k = {\sqrt{\gamma_k} \bh^{\sf H}_{k}{\bf f}_k s_k} + {\sum_{\ell \ne k,\ell = 1}^{K} \sqrt{\gamma_k}{\bf{h}}_{k}^{\sf H}{\bf f}_{\ell}s_{\ell} + n_k},
\end{align}
where $\gamma_k$ and $\bh_k \in \bbC^{N}$ are the large-scale channel fading term and small-scale channel fading vector between the AP and   user $k$, respectively, and $n_k$ is an additive white Gaussian noise (AWGN) at user $k$ with a zero mean and  variance of $\sigma^2$.
Similarly, the large-scale channel fading term and small-scale channel fading vector between the AP and eavesdroppers are represented as $\gamma_{m}^{\sf e}$ and $\bg_m \in \bbC^{N}$, respectively.
Then, the received signal at eavesdropper $m$ is given by
\begin{align}
    y^{\sf e}_{m} = \sum_{\ell = 1}^{K}\sqrt{\gamma_{m}^{\sf e}}{\bf{g}}_{m}^{\sf H}{\bf f}_{\ell}s_{\ell} + n^{\sf e}_{m},
\end{align}
where $n^{\sf e}_{m}$ is the AWGN noise at the eavesdropper $m$ with a zero mean and  variance of $\sigma^2_{\sf e}$.
We assume the perfect CSIT for both the legitimate users and eavesdroppers.
We then extend our method to a partial CSIT case.

\section{Problem Formulation} \label{sec:Prob_form}

In this section, we introduce performance metrics that incorporate the effect of FBL in the considered communication system.
The achievable rate of user $k$ is
\begin{align} 
    \label{eq:userSE}
    R_k = \log_2 \left(1 + \rho_k\right),
\end{align}
where $\rho_k$ is the signal-to-interference-plus-noise ratio (SINR) of user $k$ defined as
\begin{align}
    \rho_k = \frac{\gamma_k |{\bf{h}}_{k}^{\sf H} {\bf{f}}_{k}|^2}{ \sum_{\ell \ne k, \ell = 1}^{K} \gamma_{k} |{\bf{h}}_{k}^{\sf H} {\bf{f}}_{\ell}|^2  + \sigma^2/P}.
\end{align}
The achievable rate of eavesdropper $m$ for $s_k$ is
\begin{align}
    \label{eq:eveSE}
    R_{m,k}^{\sf e} = \log_2 \left(1 + \rho_{m,k}^{\sf e}\right),
\end{align}
where $\rho_{m,k}^{\sf e}$ is the SINR of eavesdropper $m$ for $s_k$ defined as
\begin{align}
    \rho_{m,k}^{\sf e} = \frac{\gamma^{\sf e}_{m} |{\bf{g}}^{\sf H}_{m} {\bf{f}}_{k}|^2}{\sum_{\ell \ne k, \ell = 1}^{K} \gamma^{\sf e}_{m} |{\bf{g}}^{\sf H}_{m} {\bf{f}}_{\ell}|^2  + \sigma^2_{\sf e}/P}.
\end{align}

According to \cite{PLS_FBL:ISIT:2016, PLS_FBL:TIT:2019}, the secrecy rate which measures the maximum transmission rate of the confidential information that any eavesdropper cannot decode in the FBL regime is given as
\begin{align}
    \label{eq:SecrecyRate}
    R_{k}^{\sf sec}(\bF, \epsilon_k , \delta_{m,k}; L) 
    = R_k - \sqrt{\frac{\CMcal{V}_{k}}{L}}Q^{-1}(\epsilon_k) - \max_{m \in \CMcal{M}}\left\{R^{\sf e}_{m,k} + \sqrt{\frac{\CMcal{V}^{\sf e}_{m,k}}{L}}Q^{-1}\left(\delta_{m,k}\right)\right\},
\end{align}
where $\CMcal{V}^{\sf e}_{m,k}$ and $\CMcal{V}_{k}$ are channel dispersion factors that depend on the stochastic variations of the legitimate and the wiretap channels, respectively \cite{FBL:TIT:2010, PLS_FBL:TIT:2019}, $Q^{-1}(\cdot)$ denotes an inverse Q-function defined as  $Q(x)=\frac{1}{\sqrt{2\pi}}	\int_{x}^{\infty} e^{\frac{t^2}{2}}\, dt$, $\epsilon_k$ is the decoding error probability of user $k$, and $\delta_{m,k}$ is the secrecy constraint on the information leakage of $s_k$ from  eavesdropper $m$ \cite{PLS_FBL:TIT:2019}.
Note that if the blocklength $L$ goes to infinity, then the secrecy rate in \eqref{eq:SecrecyRate} is sufficiently close to the classic secrecy SE.
Thus, back-off factors in \eqref{eq:SecrecyRate}, i.e.,
$\sqrt{\frac{\CMcal{V}_{k}}{L}}Q^{-1}(\epsilon_k)$ and $\sqrt{\frac{\CMcal{V}^{\sf e}_{m,k}}{L}}Q^{-1}\left(\delta_{m,k}\right)$, play as a drawback which reduces the secrecy rate in the FBL regime.
Considering an interference channel where  transmitters uses  an independent and identically distributed (i.i.d.) Gaussian codebook and  receivers employ nearest-neighbor decoding, the channel dispersions factors become \cite{scarlett2016dispersion, schiessl2019delay}
\begin{align}
    \label{eq:dispersion_user}
    &\CMcal{V}_{k} = \CMcal{V}^{\sf i.i.d.}(\rho_k) = \frac{2\rho_k}{1+\rho_k}\left(\log_{2} e\right)^2,
    \\
    \label{eq:dispersion_eve}
    &\CMcal{V}^{\sf e}_{m,k} = \CMcal{V}^{\sf i.i.d.}(\rho^{\sf e}_{m,k})  = \frac{2\rho^{\sf e}_{m,k}}{1+\rho^{\sf e}_{m,k}}\left(\log_{2} e\right)^2.
\end{align}

Now, we let the predetermined maximum error probability and information leakage constraints as $\hat {\epsilon}_k$ and $\hat{\delta}_{m,k}, \forall k\in \CMcal{K}, \forall m \in \CMcal{M}$, respectively.
Defining ${\boldsymbol \epsilon} = [\epsilon_1, \cdots, \epsilon_K]^{\sf T} \in \bbR^{K}_{+}$ and ${\boldsymbol \Delta} = [{\boldsymbol \delta}_{1}, \cdots, {\boldsymbol \delta}_{K}] \in \bbR^{M \times K}_{+}$, where ${\boldsymbol \delta}_{k} = [{\delta}_{1,k}, \cdots, {\delta}_{M,k}]^{\sf T} \in \bbR^{M}_{+}, \forall k \in \CMcal{K}$, we formulate a joint optimization problem to maximize the sum secrecy rate and to minimize the maximum error probability and information leakage rate as
\begin{align} 
    \label{eq:obj_func}
    &\mathop{{\text{maximize}}}_{\bF, {\boldsymbol \epsilon},{\boldsymbol \Delta}} \;\; \sum_{k = 1}^{K}\;{R}_{k}^{\sf sec}(\bF, {\boldsymbol \epsilon},{\boldsymbol \Delta}; L),
    \\
    \label{eq:err_min}
    &\mathop{{\text{minimize}}}_{{\boldsymbol \epsilon}} \;\; \max \{\epsilon_1, \cdots, \epsilon_K\},
    \\
    \label{eq:del_min}
    &\mathop{{\text{minimize}}}_{{\boldsymbol \Delta}} \;\; \max\{\delta_{1,1},\cdots,\delta_{M,K}\},
    \\
    \label{eq:power_const} 
    &{\text{subject to}} \;\; {\rm {tr}}\left(\bF \bF^{\sf H}\right) \leq 1,
    \\
    & \qquad \qquad \;\; \epsilon_k \leq \hat{\epsilon}_k,\; \forall k \in \CMcal{K},
    \\
    & \qquad \qquad \;\; \delta_{m,k} \leq \hat{\delta}_{m,k},\; \forall m \in \CMcal{M}, \forall k \in \CMcal{K},
\end{align}
where \eqref{eq:err_min} and \eqref{eq:del_min} are the error probability and information leakage rate constraints, and \eqref{eq:power_const} is a transmit power constraint at the AP, respectively.
The main challenges in solving the optimization problem are: $1)$ the problem is the multi-objective optimization, $2)$ the objective function in \eqref{eq:obj_func} is not tractable due to the maximum function which is non-smooth, $3)$ the problem is inherently non-convex, and $4)$ the error probability and information leakage need to be considered in the constraints and they are intertwined with the secrecy rate as parameters of the back-off factors.
In this regard, it is infeasible to find the global optimal solution.

\section{Proposed Joint Optimization Method} 
\label{sec:Proposed}
In this section, we develop a novel optimization method for solving the problem in \eqref{eq:obj_func} by employing an alternating optimization approach.
We first begin with obtaining the secure precoder that maximizes the sum secrecy rate for given error probabilities and information leakage rates, and then solve the minimization problem in terms of the maximum error probability and information leakage rate for the obtained precoder.

\subsection{Phase I: Sum Secrecy Rate Maximization}
\label{subsec:Phase1}
In this phase, we find the optimal precoder that maximizes the sum secrecy rate while fixing $\epsilon_k$ and $\delta_{m,k}, \forall k\in \CMcal{K}, \forall m \in \CMcal{M}$.
Since the objective function in \eqref{eq:obj_func} is not smooth, it is necessary to make \eqref{eq:SecrecyRate} smooth to find a more tractable form.
To this end, we first adopt a LogSumExp technique \cite{shen2010dual} to approximate the maximum function with a parameter $\alpha$ as
\begin{equation}\label{eq:maxLogSum}
    \max_{i=1,\cdots,N}\{x_i\} \approx \frac{1}{\alpha}\ln \left(\sum_{i=1}^N \exp(x_i \alpha) \right),
\end{equation}
where the approximation becomes tight as $\alpha \rightarrow \infty$. 
Applying \eqref{eq:maxLogSum} to \eqref{eq:SecrecyRate}, we have
\begin{align}
    &\max_{m \in \mathcal{M}}\left\{R^{\sf e}_{m,k} + \sqrt{\frac{\CMcal{V}^{\sf e}_{m,k}}{L}} Q^{-1}(\delta_{m,k}) \log_{2} e \right\} 
    \\
    &\approx \; \frac{1}{\alpha} \ln \left[\sum_{m=1}^{M} \exp\left({\alpha R^{\sf e}_{m,k}} + \alpha \sqrt{\frac{2\rho^{\sf e}_{m,k}}{L(1 + \rho^{\sf e}_{m,k})}} Q^{-1}\left(\delta_{m,k}\right) \log_{2} e \right)\right]
    \\
    &=\; \frac{1}{\alpha} \ln \left[ \sum_{m=1}^{M} \left\{ {(1+\rho^{\sf e}_{m,k})}^{\frac{\alpha}{\ln 2}}\cdot\exp{\left(\frac{\alpha}{\ln 2} \sqrt{\frac{2\rho^{\sf e}_{m,k}}{L(1+\rho^{\sf e}_{m,k})}}\; Q^{-1}\left(\delta_{m,k}\right) \right)} \right\} \right]
    \\
    \label{eq:LogSumExP}
    &=\; \Tilde{R}^{\sf e}_{k}.
\end{align}
Now, we introduce the following lemma \cite{choi:iotj:21}:
\begin{lemma}
    \label{lem:lowerbound}
    For any given $x$, $\Tilde{x}>0$, an upper bound of $\sqrt{2x/(1+x)}$ is obtained as
    \begin{align}
        \sqrt{\frac{2x}{1+x}} \leq q(\Tilde{x})\ln (1+x) + r(\Tilde{x}),
    \end{align}
    where $q(\Tilde{x}) = \frac{1}{\sqrt{2\Tilde{x}(1+\Tilde{x})}}$ and $r(\Tilde{x}) = \sqrt{\frac{2\Tilde{x}}{1+\Tilde{x}}} - q(\Tilde{x})\ln (1+\Tilde{x})$.
\end{lemma}
\begin{proof}
Refer to the proof of Lemma 2 in \cite{choi:iotj:21}.
\end{proof}
Based on Lemma~\ref{lem:lowerbound} with \eqref{eq:LogSumExP}, we can obtain the lower bound of the approximated secrecy rate.
    For given $\tilde{\rho}_k$ and $\tilde{\rho}^{\sf e}_{m,k}$, the lower bound of \eqref{eq:SecrecyRate} is obtained as
    \begin{align}
        {R}^{\sf sec}_{k} \overset{(a)}{\approx}&\; R_k -  \sqrt{\frac{2\rho_k}{L(1 + \rho_k)}} Q^{-1}(\epsilon_k)\log_{2}e - \Tilde{R}^{\sf e}_{k} 
        \\
        \overset{(b)}{\geq}&\; \log_2 (1+\rho_k) - \frac{Q^{-1} (\epsilon_k)}{\sqrt{L}}q(\tilde{\rho}_{k})\ln (1+\rho_k)\log_2 e - \psi_k
        \nonumber \\
        & -\frac{1}{\alpha}\ln \left\{\sum_{m=1}^{M} {(1+\rho^{\sf e}_{m,k})}^{\frac{\alpha}{\ln 2}}\cdot\exp{\left(\ln (1+\rho^{\sf e}_{m,k})^{\frac{\alpha Q^{-1}(\delta_{m,k})}{\sqrt{L}\ln 2}q(\tilde{\rho}^{\sf e}_{m,k})} + \alpha \psi_{m,k}^{\sf e}\right)}\right\}
        \\
        =&\; \log_2 (1+\rho_k)^{\omega_k} - \psi_k - \frac{1}{\alpha} \ln \left\{\sum_{m=1}^{M} {\color{black}\beta_{m,k}} (1 + \rho^{\sf e}_{m,k})^{\omega^{\sf e}_{m,k}}\right\} 
        \\
        =&\; {R}^{\sf sec,\sf lb}_{k}, \label{eq:lb_Rmk}
    \end{align}
    where $(a)$ comes from \eqref{eq:LogSumExP}, $(b)$ follows from Lemma~\ref{lem:lowerbound}, and
    \begin{align}
        &\psi_k = \frac{Q^{-1}(\epsilon_k) \log_2 e}{\sqrt{L}}r(\tilde{\rho}_k),
        \; \psi_{m,k}^{\sf e} = \frac{Q^{-1}(\delta_{m,k}) \log_2 e}{\sqrt{L}}r(\tilde{\rho}_{m,k}^{\sf e}), \;\beta_{m,k} = \exp{\left(\alpha \psi^{\sf e}_{m,k} \right)},
        \\
        &\omega_k =\; 1 - \frac{Q^{-1}(\epsilon_k)}{\sqrt{L}}q(\tilde{\rho}_{k}),
    \text{ and } \omega_{m,k}^{\sf e} =\; \frac{\alpha}{\ln{2}}\left(1 + \frac{Q^{-1}\left(\delta_{m,k}\right)}{\sqrt{L}} q(\tilde{\rho}^{\sf e}_{m,k})\right).
    \end{align}
Since we pursue to solve \eqref{eq:obj_func} for given $\boldsymbol{\epsilon}$ and $\boldsymbol{\Delta}$ with the lower bound in \eqref{eq:lb_Rmk}, our problem is transformed to the single-objective maximization problem as
\begin{align} 
    \label{eq:obj_func_LB}
    \mathop{{\text{maximize}}}_{\bF}& \;\; \sum_{k = 1}^{K}\;{R}^{\sf sec,\sf lb}_{k}(\bF),
    \\
    \label{eq:power_const2} 
    {\text{subject to}} & \;\;{\rm {tr}}\left(\bF \bF^{\sf H}\right) \leq 1.
\end{align}

Next, to further obtain a compact rate expression with respect to the precoder, we vectorize the precoding matrix as
\begin{align} 
    \label{eq:precodingvector}
    \bar{\mathbf{f}} = {\rm{vec}}(\bF) = \left[\mathbf{f}^\top_{1},\mathbf{f}^\top_{2},\ldots,\mathbf{f}^\top_{K}\right]^\top \in \mathbb{C}^{NK}.
\end{align}
We set $\|\bar{\bf f}\|^2 = 1$, i.e., transmission with the maximum power, because the secrecy rate increases with the transmit power.
Then, we can finally reformulate the problem in  \eqref{eq:obj_func_LB} as the product of Rayleigh quotients:
\begin{align} 
    \label{eq:Object_GPI}
    \mathop{{\text{maximize}}}_{\bar{\mathbf{f}}}\; \sum_{k=1}^{K}\left[\log_2 \left(\frac{\bar{\mathbf{f}}^{\sf H}\mathbf{A}_{k}\bar{\mathbf{f}}}{\bar{\mathbf{f}}^{\sf H}\mathbf{B}_{k}\bar{\mathbf{f}}}\right)^{\omega_k} - \ln \left\{\sum_{m=1}^{M} {\color{black} \beta_{m,k}} \left(\frac{\bar{\mathbf{f}}^{\sf H}\mathbf{C}_{m}\bar{\mathbf{f}}}{\bar{\mathbf{f}}^{\sf H}\mathbf{D}_{m,k}\bar{\mathbf{f}}}\right)^{\omega_{m,k}^{\sf e}}\right\}^{ \frac{1}{\alpha}}\right]\!,
\end{align}
where
\begin{align}
    \label{eq:A_k}
    &\mathbf{A}_k = \mathrm{blkdiag}\left(\gamma_k \mathbf{h}_{k} \mathbf{h}^{\sf H}_{k},\cdots, \gamma_k\mathbf{h}_{k} \mathbf{h}^{\sf H}_{k}\right) + \mathbf{I}_{NK}\frac{\sigma^2}{P} \in \mathbb{C}^{NK\times NK},
    \\
    \label{eq:B_k}
    &\mathbf{B}_k = \mathbf{A}_k - \mathrm{blkdiag}(\mathbf{0},\cdots, \gamma_k \mathbf{h}_{k}\mathbf{h}^{\sf H}_{k},\cdots, \mathbf{0}) \in \mathbb{C}^{NK\times NK},
    \\
    \label{eq:C_k}
    &\mathbf{C}_{m} = \mathrm{blkdiag}\left(\gamma^{\sf e}_{m} \bg_{m} \bg_{m}^{\sf H},\cdots, \gamma^{\sf e}_{m} \bg_{m} \bg_{m}^{\sf H}\right) + \mathbf{I}_{NK}\frac{\sigma^2_{\sf e}}{P} \in \mathbb{C}^{NK\times NK},
    \\
    \label{eq:D_mk}
    &\mathbf{D}_{m,k} = \mathbf{C}_{m} - \mathrm{blkdiag}(\mathbf{0},\cdots, \gamma^{\sf e}_{m} \bg_{m} \bg_{m}^{\sf H},\cdots, \mathbf{0}) \in \mathbb{C}^{NK\times NK}.
\end{align}
The second terms in both \eqref{eq:B_k} and \eqref{eq:D_mk} have nonzero blocks which are located at the $k$th diagonal block.
We note that the product of Rayleigh quotient forms in \eqref{eq:Object_GPI} is derived under the assumption $\|\bar{{\bf f}}\|^2=1$, and the problem in \eqref{eq:Object_GPI} is invariant up to the scaling of $\bar{\bf f}$. 
Accordingly, the power constraint in \eqref{eq:power_const2} is removed in the reformulated problem.

Now, we focus on identifying the local points of the problem in \eqref{eq:Object_GPI}. 
For simplicity, we define the objective function in \eqref{eq:Object_GPI} as
\begin{align} 
    \label{eq:log_obj}
    \cL_1(\bar{\bf f}) =&\; \log_2\prod_{k=1}^{K} \left[\left(\frac{\bar{\mathbf{f}}^{\sf H}\mathbf{A}_{k}\bar{\mathbf{f}}}{\bar{\mathbf{f}}^{\sf H}\mathbf{B}_{k}\bar{\mathbf{f}}}\right)^{\omega_k}\!\left\{\sum_{m=1}^{M} {\color{black} \beta_{m,k}} \left(\frac{\bar{\mathbf{f}}^{\sf H}\mathbf{C}_{m}\bar{\mathbf{f}}}{\bar{\mathbf{f}}^{\sf H}\mathbf{D}_{m,k}\bar{\mathbf{f}}}\right)^{\omega_{m,k}^{\sf e}}\right\}^{-\frac{\ln 2}{\alpha}}\right]
    \\
    \label{eq:lamb}
    =&\; \log_2 \lambda(\bar{\bf f}).
\end{align}
Then, we derive Lemma~\ref{lem:NEP} to find the condition of stationary points of \eqref{eq:log_obj}.
\begin{lemma}
    \label{lem:NEP}
    The first-order KKT condition of the problem \eqref{eq:Object_GPI} is satisfied if
    \begin{align}
        \label{eq:general_eigen}
        \mathbf{B}^{-1}_{\sf KKT}(\bar{\mathbf{f}})\mathbf{A}_{\sf KKT}(\bar{\mathbf{f}})\bar{\mathbf{f}} = \lambda(\bar{\mathbf{f}})\bar{\mathbf{f}},
    \end{align}
where 
    \begin{align}
        \label{eq:A_KKT}
        &\mathbf{A}_{\sf KKT}(\bar{\mathbf{f}}) \!=\! \sum_{k=1}^{K}\left[\frac{\omega_k}{\ln 2}\left(\frac{\mathbf{A}_k}{\bar{\mathbf{f}}^{\sf H}\mathbf{A}_{k}\bar{\mathbf{f}}}\right) \!+\! \frac{1}{\alpha}\left\{\sum_{m=1}^{M}\left(\frac{\omega^{\sf e}_{m,k} {\color{black} \beta_{m,k}} \left(\frac{\bar{\mathbf{f}}^{\sf H}\mathbf{C}_{m}\bar{\mathbf{f}}}{\bar{\mathbf{f}}^{\sf H}\mathbf{D}_{m,k}\bar{\mathbf{f}}}\right)^{\omega_{m,k}^{\sf e}} \frac{\mathbf{D}_{m,k}}{\bar{\mathbf{f}}^{\sf H}\mathbf{D}_{m,k}\bar{\mathbf{f}}}}{\sum_{\ell =1}^{M} {\color{black} \beta_{m,k}}\left(\frac{\bar{\mathbf{f}}^{\sf H}\mathbf{C}_{\ell}\bar{\mathbf{f}}}{\bar{\mathbf{f}}^{\sf H}\mathbf{D}_{\ell,k}\bar{\mathbf{f}}} \right)^{\omega_{\ell,k}^{\sf e}}}\right)\right\}\right]\lambda_{\sf num}(\bar{\mathbf{f}}), \\
        \label{eq:B_KKT}
        &\mathbf{B}_{\sf KKT}(\bar{\mathbf{f}}) \!=\! \sum_{k=1}^{K}\left[\frac{\omega_k}{\ln 2}\left(\frac{\mathbf{B}_k}{\bar{\mathbf{f}}^{\sf H}\mathbf{B}_{k}\bar{\mathbf{f}}}\right) \!+\! \frac{1}{\alpha}\left\{\sum_{m=1}^{M}\left(\frac{\omega^{\sf e}_{m,k} {\color{black} \beta_{m,k}} \left(\frac{\bar{\mathbf{f}}^{\sf H}\mathbf{C}_{m}\bar{\mathbf{f}}}{\bar{\mathbf{f}}^{\sf H}\mathbf{D}_{m,k}\bar{\mathbf{f}}}\right)^{\omega_{m,k}^{\sf e} }\frac{\mathbf{C}_{m}}{\bar{\mathbf{f}}^{\sf H}\mathbf{C}_{m}\bar{\mathbf{f}}}}{\sum_{\ell=1}^{M} {\color{black} \beta_{m,k}}
        \left(\frac{\bar{\mathbf{f}}^{\sf H}\mathbf{C}_{\ell}\bar{\mathbf{f}}}{\bar{\mathbf{f}}^{\sf H}\mathbf{D}_{\ell,k}\bar{\mathbf{f}}} \right)^{\omega_{\ell,k}^{\sf e}}}\right)\right\}\right]\lambda_{\sf den}(\bar{\mathbf{f}}),
        \\
        \label{eq:lamb_num}
        &\lambda_{\sf num}(\bar{\mathbf{f}}) = \prod_{k=1}^{K} \left(\frac{\bar{\mathbf{f}}^{\sf H}\mathbf{A}_{k}\bar{\mathbf{f}}}{\bar{\mathbf{f}}^{\sf H}\mathbf{B}_{k}\bar{\mathbf{f}}}\right)^{\omega_k},\;
        \lambda_{\sf den}(\bar{\mathbf{f}}) = \prod_{k=1}^{K} \left\{\sum_{m=1}^{M} {\color{black} \beta_{m,k}} \left(\frac{\bar{\mathbf{f}}^{\sf H}\mathbf{C}_{m}\bar{\mathbf{f}}}{\bar{\mathbf{f}}^{\sf H}\mathbf{D}_{m,k}\bar{\mathbf{f}}}\right)^{\omega_{m,k}^{\sf e}}\right\}^{\frac{\ln 2}{\alpha }}.
    \end{align}
\end{lemma}
\begin{proof}
See Appendix \ref{pf:KKT_condition}. 
\end{proof}

We note that the first-order KKT condition in \eqref{eq:general_eigen} can be interpreted as a generalized eigenvalue problem $\mathbf{B}^{-1}_{\sf KKT}(\bar{\mathbf{f}})\mathbf{A}_{\sf KKT}(\bar{\mathbf{f}})\bar{\mathbf{f}} = \lambda(\bar{\mathbf{f}})\bar{\mathbf{f}}$.
Here, $\lambda(\bar{\bf f})$ is  as an eigenvalue of ${\bf{B}}_{\sf KKT}^{-1}(\bar {\bf{f}}) {\bf{A}}_{\sf KKT}(\bar {\bf{f}})$ with $\bar{\bf f}$ as a corresponding eigenvector.
As a result, maximizing the objective function $\cL_1 (\bar{\bf f})$ is equivalent to maximizing $\lambda(\bar{\bf f})$.
Therefore, it is desirable to find the principal eigenvalue of \eqref{eq:general_eigen} to maximize \eqref{eq:lamb}, which is equivalent to finding the best local optimal solution of \eqref{eq:Object_GPI}.

\begin{algorithm} [t]
\caption{Sum Secrecy Rate Maximization} 
\label{alg:algorithm_1} 
{\bf{initialize}}: $\bar {\bf{f}}^{(0)}$, $t = 1$.
\\
\While {$\left\|\bar {\bf{f}}^{(t)} - \bar {\bf{f}}^{(t-1)} \right\| > \varepsilon$ \& $t \leq t_{\max} $}{
Build $ {\bf{A}}_{\sf KKT} (\bar {\bf{f}}^{(t-1)})$ and ${\bf{B}}_{\sf KKT} (\bar {\bf{f}}^{(t-1)})$ according to \eqref{eq:A_KKT} and \eqref{eq:B_KKT} for given $\epsilon_k$ and $\delta_{m,k}$.
\\
Compute $\bar {\bf{f}}^{(t)} = {\bf{B}}^{-1}_{\sf KKT} (\bar {\bf{f}}^{(t-1)}) {\bf{A}}_{\sf KKT} (\bar {\bf{f}}^{(t-1)}) \bar {\bf{f}}^{(t-1)}$. 
\\
Normalize $\bar {\bf{f}}^{(t)} = \bar {\bf{f}}^{(t)}/\left\| \bar {\bf{f}}^{(t)}\right\|$.
\\
 $t \leftarrow t+1$.}
$\bar {\bf{f}}^\star \leftarrow \bar {\bf{f}}^{(t)}$
\\
\Return{\ }{$\bar{{\bf f}}^{\star} = \left[{\bf{f}}_1^{\sf T}, {\bf{f}}_2^{\sf T}, \dots, {\bf{f}}_K^{\sf T} \right]^{\sf T} $}.
\end{algorithm}
Based on \eqref{eq:general_eigen}, we propose the sum secrecy rate maximization precoding algorithm by adopting the generalized power iteration (GPI) method  \cite{choi:twc:20}. 
As described in Algorithm~\ref{alg:algorithm_1}, we initialize  $\bar \bff^{(0)}$ and update $\bar \bff^{(t)}$ at each iteration; with the given $L$, $\boldsymbol{\epsilon}$, and $\boldsymbol{\Delta}$, the algorithm builds  ${\bf{A}}_{\sf KKT} (\bar {\bf{f}}^{(t-1)})$ and ${\bf{B}}_{\sf KKT} (\bar {\bf{f}}^{(t-1)})$ according to \eqref{eq:A_KKT} and \eqref{eq:B_KKT}. 
Then, the algorithm updates $\bar \bff^{(t)}$  by computing $\bar {\bf{f}}^{(t)} = {\bf{B}}^{-1} _{\sf KKT} (\bar {\bf{f}}^{(t-1)}){\bf{A}}_{\sf KKT} (\bar {\bf{f}}^{(t-1)}) \bar {\bf{f}}^{(t-1)}$ and normalizing as $\bar {\bf{f}}^{(t)} = \bar {\bf{f}}^{(t)}/ \left\| \bar {\bf{f}}^{(t)} \right\|$.
We repeat these steps until either $\bar {\bf{f}}^{(t)}$ converges to a tolerance level $\varepsilon$ or the algorithm reaches  $t_{\rm max}$.

The complexity of Algorithm~\ref{alg:algorithm_1} is dominated by the inversion in ${\bf{B}}_{\sf KKT}^{-1} (\bar {\bf{f}})$.
Since ${\bf{B}}_{\sf KKT} (\bar {\bf{f}})$ is a block-diagonal and symmetric matrix, 
we implement the inversion by exploiting a divide-and-conquer method so that we only need a computational complexity order of $\CMcal{O}\left(\frac{1}{3}K N^3\right)$ \cite{choi:twc:20} instead of $\CMcal{O}\left(K^3 N^3\right)$.
In this regard, the total complexity is $\CMcal{O} \left(\frac{1}{3}TKN^3\right)$, where $T$ is the number of iterations.
We note that the complexity of Algorithm~\ref{alg:algorithm_1} is same as a representative low-complexity framework in sum rate maximization (not even secrecy rate), namely, the weighted minimum mean square error method \cite{chris:twc:08}. 
In \cite{mei2017artificial}, a AN-aided transmission method was proposed based on a sequence of semi-definite programs (SDP) needs a computational complexity order of $\CMcal{O}\left(N^{6.5}\right)$ for a single confidential message.
In \cite{shi2020artificial}, the SDP-based algorithm was also proposed  for a single user and multiple eavesdroppers with a complexity order of $\CMcal{O} \left(N^{8.5}\right)$.
In addition, a jamming noise-aided precoding algorithm was developed 
for multiple users and eavesdroppers, which needs a  complexity order of $\CMcal{O} \left((N+K)^3\right)$ \cite{nguyen2016joint}.
Thus, we emphasize that the complexity of Algorithm~\ref{alg:algorithm_1} is low compared to the existing algorithms.
\begin{remark} 
[Secure Precoding in the FBL Regime]\normalfont
Since Algorithm~\ref{alg:algorithm_1} maximizes the secrecy rate in the FBL regime under given $\boldsymbol{\epsilon}$ and $\boldsymbol{\Delta}$ with low complexity, Algorithm~\ref{alg:algorithm_1} is considered to be a computationally  efficient secure precoding algorithm in the FBL-based communications.
\end{remark} 

\subsection{Phase II: Maximum Error Probability and Information Leakage Rate Minimization}
\label{subsec:Phase2}
We find the optimal ${\boldsymbol \epsilon}$ and ${\boldsymbol \Delta}$ for fixed $\bF$.
To this end, we transform the multi-objective problem to a single-objective form and  derive a solution
for the transformed problem.
Adopting the weighted-sum approach \cite{FBL:WCL:2017}, the multi-objective problem in \eqref{eq:obj_func} for given $\bF$ is converted to
\begin{align}
    \label{eq:multi_obj}
    &\mathop{{\text{maximize}}}_{{\boldsymbol \epsilon}, {\boldsymbol \Delta}} \;\; \frac{w}{R_{\infty}} \sum_{k=1}^{K} {R}_{k}^{\sf sec}(\boldsymbol \epsilon , \boldsymbol \Delta) \!+\! (1-w)\left(\frac{\hat{\epsilon}_{\sf max} - {\rm max}\{{\boldsymbol \epsilon}\}}{\hat{\epsilon}_{\sf max}} + \frac{\hat{\delta}_{\sf max} - {\rm max}\{{\boldsymbol \Delta}\}}{\hat{\delta}_{\sf max}}\right),
    \\
    &{\text{subject to}} \;\; \epsilon_k \leq \hat{\epsilon}_k,\; \forall k \in \CMcal{K},
    \\
    & \qquad \qquad \;\; \delta_{m,k} \leq \hat{\delta}_{m,k},\; \forall m \in \CMcal{M}, \forall k \in \CMcal{K},
\end{align}
where $R_{\infty}$ indicates a normalization constant that can be obtained as the sum secrecy rate in the infinite blocklength regime computed with a state-of-the-art secure precoder, $\hat{\epsilon}_{\sf max} = {\rm max}\{\hat{\epsilon}_1 ,\cdots, \hat{\epsilon}_K\}$, $\hat{\delta}_{\sf max} = {\rm max}\{\hat{\delta}_{1,1} ,\cdots, \hat{\delta}_{M,K}\}$, and the weight is $w \in [0, 1]$.
To solve \eqref{eq:multi_obj}, we need to handle the maximum function in $R_k^{\sf sec}(\epsilon_k,\delta_{m,k})$.
We consider the upper bound by adding the wiretap rates as $\max_{m \in \CMcal{M}}\left\{R^{\sf e}_{m,k} + \sqrt{\frac{\CMcal{V}_{m,k}^{\sf e}}{L}}Q^{-1}\left(\delta_{m,k}\right)\right\} \leq \sum_{m = 1}^{M} \left( R^{\sf e}_{m,k} +  \sqrt{\frac{\CMcal{V}^{\sf e}_{m,k}}{L}}Q^{-1}\left(\delta_{m,k}\right) \right)$.
We note that although the bound  is not tight and the gap increases with the number of eavesdroppers, the proposed method based on the bound will show high performance for many eavesdroppers in Section~\ref{subsec:Phase1}.
Subsequently, letting $\tau = {\rm max}\{{\boldsymbol \epsilon}\}$ and $\xi = {\rm max}\{{\boldsymbol \Delta}\}$ and removing the terms that are not a function of either $\epsilon_k$ or $\delta_{m,k}$, the problem in \eqref{eq:multi_obj} is further transformed to the minimization problem as 
\begin{align}
    \label{eq:Error_obj}
    &\mathop{{\text{minimize}}}_{{\boldsymbol \epsilon}, {\boldsymbol \Delta}} \;\; \frac{w}{R_{\infty}} \sum_{k=1}^{K} \left[\sqrt{\frac{\CMcal{V}_k}{L}}Q^{-1}(\epsilon_k)
    + \sum_{m = 1}^{M} \left(\sqrt{\frac{\CMcal{V}^{\sf e}_{m,k}}{L}}Q^{-1}\left(\delta_{m,k}\right) \right)\right]
    +(1-w)\left(\frac{\tau}{\hat{\epsilon}_{\sf max}} + \frac{\xi}{\hat{\delta}_{\sf max}}\right),
    \\
    &{\text{subject to}} 
    \label{eq:constr_ep_1}
    \;\; \epsilon_k \leq \tau,
    \\
    \label{eq:constr_ep_2}
    &\qquad \qquad \;\; 0 \leq \epsilon_k \leq \hat{\epsilon}_{k},
    \\
    \label{eq:constr_ep_3}
    & \qquad \qquad \;\; 0 \leq \tau \leq \hat{\epsilon}_{\sf max},
    \\ 
    \label{eq:constr_del_1}
    & \qquad \qquad \;\; \delta_{m,k} \leq \xi,
    \\
    \label{eq:constr_del_2}
    &\qquad \qquad \;\; 0 \leq \delta_{m,k} \leq \hat{\delta}_{m,k},
    \\
    \label{eq:constr_del_3}
    & \qquad \qquad \;\; 0 \leq \xi \leq \hat{\delta}_{\sf max},
\end{align}
where \eqref{eq:constr_ep_2} and \eqref{eq:constr_del_2} are the error probability and information leakage  rate constraints, and
\eqref{eq:constr_ep_3} and \eqref{eq:constr_del_3} indicate the maximum error probability assumption and information leakage assumptions.
We note that the first term in \eqref{eq:Error_obj} is monotonically increasing while the second term in \eqref{eq:Error_obj} is monotonically decreasing with decreasing error probability and information leakage.
Therefore, the optimal tradeoff exists among the back-off factors, error probability, and information leakage rate which is found by solving the KKT conditions \cite{Boyd:Convex:2004} of \eqref{eq:Error_obj}.
\begin{lemma}
    \label{lem:ErrorMin}
    The optimal upper decoding error probability and information leakage rate for given $\ell, j(k) \in \bbN_{+}$ are derived as
    \begin{align}
        \label{eq:tau_star}
        &\tau^{\star} = Q\left(\sqrt{2 \ln \left(\frac{\sqrt{L}(1-w)R_{\infty}}{\hat{\epsilon}_{\sf max}w\sqrt{2\pi}\sum_{k=\ell}^{K}\sqrt{\CMcal{V}_k}}\right)}\right),
        \\
        \label{eq:xi_star}
        &\xi^{\star} = Q\left(\sqrt{2 \ln \left(\frac{\sqrt{L}(1-w)R_{\infty}}{\hat{\delta}_{\sf max}w\sqrt{2\pi}\sum_{m=j(k)}^{M}\sum_{k=1}^{K}\sqrt{\CMcal{V}_{m,k}^{\sf e}}}\right)}\right).
    \end{align}
    Then, without loss of generality, we can assume that $\hat{\epsilon}_{\ell-1} < \tau < \hat{\epsilon}_{\ell}$ for some $\ell$ and $\hat{\delta}_{j(k)-1,k} < \xi < \hat{\delta}_{j(k),k}$ for some $j(k)$. 
    Then, the optimal solution of the problem in \eqref{eq:Error_obj} is
    \begin{align}
        &{\boldsymbol{\epsilon}}^{\star} = [\hat{\epsilon}_1, \cdots, \hat{\epsilon}_{\ell -1},\tau^{\star},\cdots,\tau^{\star}]^{\sf T} \in \bbR_{+}^{K},
        \\
        &{\boldsymbol{\Delta}}^{\star} = [{\boldsymbol{\delta}}_{1}^{\star}, \cdots, {\boldsymbol{\delta}}_{K}^{\star}] \in \bbR_{+}^{M \times K},
    \end{align}
    where 
    \begin{align}
        {\boldsymbol{\delta}}_{k}^{\star} = [\hat{\delta}_{1,k}, \cdots, \hat{\delta}_{j(k) -1,k},\xi^{\star},\cdots,\xi^{\star}]^{\sf T} \in \bbR_{+}^{M},\; \forall k \in \CMcal{K}.
    \end{align}
    If $\hat{\epsilon}_K < \tau^{\star}$, the optimal error probability becomes
    \begin{align}
        {\boldsymbol{\epsilon}}^{\star} = [\hat{\epsilon}_1, \cdots, \hat{\epsilon}_K]^{\sf T}.
    \end{align}
    Similarly, if $\hat{\delta}_{M,k} < \xi^{\star}, \forall k \in \CMcal{K}$, the optimal information leakage rate becomes
    \begin{align}
        {\boldsymbol{\delta}}^{\star}_k = [\hat{\delta}_{1,k}, \cdots, \hat{\delta}_{M,k}]^{\sf T}, \; \forall k \in \CMcal{K}.
    \end{align}
\end{lemma}
\begin{proof}
    See Appendix~\ref{pf:ErrorMin}.
\end{proof}

\begin{remark}
    [Optimality of Solutions]\normalfont
    The derived solutions in Lemma~\ref{lem:ErrorMin} are the optimal error probability solution and suboptimal information leakage rate solution  for the problem in \eqref{eq:multi_obj}.
    For a special case $m=1$, both the derived solutions are optimal.
\end{remark}

\subsection{Joint Secure Precoding Algorithm}
\label{subsec:proposed_algorithm}

\begin{algorithm} [t]
\caption{Joint Optimization based on Alternating Approach} 
\label{alg:algorithm_2} 
{\bf{initialize}}: $\bar {\bf{f}}^{(0)},{\boldsymbol{\epsilon}}^{(0)}, {\boldsymbol{\delta}}^{(0)}_{k}, \forall k \in \CMcal{K}$, and  $t = 1$.
\\
\While {$increment\; of\;  \eqref{eq:multi_obj} > \varepsilon^{\rm out}$ \&  $t\leq t_{\rm max}^{\rm out}$}
{
$\bar {\bf{f}}^{(t)} \leftarrow \text{Algorithm~\ref{alg:algorithm_1}}\left( {\boldsymbol{\epsilon}}^{(t-1)}, {\boldsymbol{\Delta}}^{(t-1)}\right)$.
\\
Compute $R_{\infty}$ by using Algorithm~\ref{alg:algorithm_1} with $L = \infty$.
\\
Find $\tau^{\star}$ and $\xi^{\star}$ according to \eqref{eq:tau_star} and \eqref{eq:xi_star} for $\bar{{\bf{f}}}^{(t)}$.
\\
Set ${\boldsymbol{\epsilon}}^{(t)} \!=\! [\hat{\epsilon}_1, \cdots, \hat{\epsilon}_{\ell -1},\tau^{\star},\cdots,\tau^{\star}]^{\sf T}$ and ${\boldsymbol{\delta}}_{k}^{(t)} \!=\! [\hat{\delta}_{1,k}, \cdots, \hat{\delta}_{j(k) -1,k},\xi^{\star},\cdots,\xi^{\star}]^{\sf T}$ $\forall k$.
\\
$t \leftarrow t+1$.
}
$\bar {\bf{f}}^\star \leftarrow \bar {\bf{f}}^{(t)}$, ${\boldsymbol{\epsilon}}^{\star} \leftarrow {\boldsymbol{\epsilon}}^{(t)}$, and ${\boldsymbol{\Delta}}^{\star} \leftarrow {\boldsymbol{\Delta}}^{(t)}$.
\\
\Return{\ }{$\bar{{\bf f}}^{\star} = \left[{\bf{f}}_1^{\sf T}, {\bf{f}}_2^{\sf T}, \dots, {\bf{f}}_K^{\sf T} \right]^{\sf T} $}, ${\boldsymbol{\epsilon}}^{\star}$, and ${\boldsymbol{\Delta}}^{\star}$.
\end{algorithm}

From the results in Section~\ref{subsec:Phase1} and Section~\ref{subsec:Phase2}, we finally design our proposed algorithm, which is described in Algorithm~\ref{alg:algorithm_2}.
First, Algorithm~\ref{alg:algorithm_2} initializes the precoding vector $\bar \bff^{(0)}$, ${\boldsymbol{\epsilon}}^{(0)}$, and ${\boldsymbol{\delta}}^{(0)}_{k}, \forall k \in \CMcal{K}$.
To find the best local optimal precoding vector $\bar{\bf f}^{(t)}$, Algorithm~\ref{alg:algorithm_1} is used for given ${\boldsymbol{\epsilon}}^{(t-1)}$ and ${\boldsymbol{\delta}}^{(t-1)}_{k}$.
Then, $R_{\infty}$ is computed by finding $\bar{\bf f}$ from Algorithm~\ref{alg:algorithm_1} in the infinite blocklength regime.
The optimal $\tau^{\star}$ and $\xi^{\star}$ are computed based on \eqref{eq:tau_star} and \eqref{eq:xi_star}, and we set ${\boldsymbol{\epsilon}}^{(t)}$ and ${\boldsymbol{\Delta}}^{(t)}$ accordingly.
We repeat these steps until either the objective function
in \eqref{eq:multi_obj} increases smaller than $\varepsilon^{\rm out}$ compared to the previous iteration, where $\varepsilon^{\rm out}>0$ denotes a tolerance threshold for the outer loop, or the algorithm reaches a maximum iteration count $t_{\rm max}^{\rm out}$.
Since the solutions in \eqref{eq:tau_star} and \eqref{eq:xi_star} are closed form, the complexity order of Algorithm~\ref{alg:algorithm_2} is  $\CMcal{O} \left(\frac{1}{3}T_{\rm tot} KN^3\right)$, where $T_{\rm tot}$ is the number of  total iterations of Algorithm~\ref{alg:algorithm_1}.

\section{Extension to Partial CSIT of Wiretap Channels}
\label{sec:imperfect_csi}

Since it is hard to obtain the perfect channel state information at the transmitter (CSIT) of eavesdroppers, we consider the case of the partial CSIT of wiretap channels in which only the long-term  channel statistics of eavesdroppers, i.e., the channel covariance $\bR_{m}^{\sf e}$, is available.
Since the instantaneous CSI of eavesdroppers is not known, it is infeasible to consider the instantaneous wiretap rate.
Accordingly, we consider the ergodic wiretap rate to exploit the partial CSIT.
Then, instead of the secrecy rate in \eqref{eq:SecrecyRate}, we have 
\begin{align}
    \label{eq:UB_Imperfect}
      R_{k}^{\sf sec,\sf p} = R_k - \sqrt{\frac{\CMcal{V}_{k}}{L}}Q^{-1}(\epsilon_k) - \max_{m \in \CMcal{M}}\left\{\mathbb{E}_{\bg}\left[R^{\sf e}_{m,k} + \sqrt{\frac{\CMcal{V}^{\sf e}_{m,k}}{L}}Q^{-1}\left(\delta_{m,k}\right)\right]\right\},
\end{align}
where $\mathbb{E}_{\bg}[\cdot]$ indicates the expectation with respect to wiretap channels.
Then, we can rewrite \eqref{eq:UB_Imperfect} by using the following proposition.
\begin{proposition}
The approximated upper bound of the ergodic wiretap secrecy rate in \eqref{eq:UB_Imperfect} is obtained as
\begin{align}
    \label{eq:UB_Imperfect_2}
     \mathbb{E}_{\bg}\left[R_{m, k}^{\sf e} + \sqrt{\frac{{\CMcal{V}}^{\sf e}_{m,k}}{L}}Q^{-1}\left(\delta_{m,k}\right)\right]  \lessapprox  \bar R_{m, k}^{\sf e} + \sqrt{\frac{\bar{\CMcal{V}}^{\sf e}_{m,k}}{L}}Q^{-1}\left(\delta_{m,k}\right),
\end{align}
where
\begin{align}
    &\bar R_{m, k}^{\sf e} = \log_2 \left( 1 + \frac{\gamma^{\sf e}_{m}\bff^{\sf H}_{k} {\bR}^{\sf e}_{m} \bff_{k}}{\sum^{K}_{\ell \neq k} \gamma^{\sf e}_{m}{\bff}^{\sf H}_{\ell} {\bR}^{\sf e}_{m} {\bff}_{\ell} + {\sigma}^{2}_{\sf e}/P } \right),
    \\
    &\bar{\CMcal{V}}^{\sf e}_{m,k} = \sqrt{\frac{2 \gamma^{\sf e}_{m}\bff^{\sf H}_{k} {\bR}^{\sf e}_{m} \bff_{k}}{\sum^{K}_{k = 1} \gamma^{\sf e}_{m}{\bff}^{\sf H}_{\ell} {\bR}^{\sf e}_{m} {\bff}_{\ell} + {\sigma}^{2}_{\sf e}/P}} \log_2 e.
\end{align}
\end{proposition}
\begin{proof}
The ergodic wiretap rate in \eqref{eq:eveSE} is approximated as
\begin{align} 
    \mathbb{E}_{\bg}[R_{m,k}^{\sf e} ] =&\;\mathbb{E}_{\bg}\left[\log_2\left(1 \!+\! \frac{\gamma^{\sf e}_{m}|\bg_{m}^{\sf H} {\bf{f}}_k|^2}{\sum_{\ell \neq k}^{K}\gamma^{\sf e}_{m}|\bg_{m}^{\sf H} {\bf{f}}_\ell|^2 + \sigma_{\sf e}^2/P}\right)\right]
    \overset{(a)}{\approx} \log_2 \left( 1\! +\! \frac{\gamma^{\sf e}_{m}\bff^{\sf H}_{k} {\bR}^{\sf e}_{m} \bff_{k}}{\sum^{K}_{\ell \neq k} \gamma^{\sf e}_{m}{\bff}^{\sf H}_{\ell} {\bR}^{\sf e}_{m} {\bff}_{\ell} + {\sigma}^{2}_{\sf e}/P } \right)
    \\
    =&\;\bar{R}_{m, k}^{\sf e},
\end{align}
where $(a)$ comes from Lemma 1 in \cite{zhang2014power}. 
We also have
\begin{align}
     \mathbb{E}\left[\sqrt{{\CMcal{V}}^{\sf e}_{m,k}}\right] \overset{(b)}{\leq}&\;
     \sqrt{\mathbb{E}\left[{\CMcal{V}}^{\sf e}_{m,k}\right]}
     \overset{(c)}{\leq} \sqrt{\frac{2}{\frac{1}{\mathbb{E}\left[\rho^{\sf e}_{m,k}\right]} + 1}}\log_2 e
     \overset{(d)}{\approx} \sqrt{\frac{2 \gamma^{\sf e}_{m}\bff^{\sf H}_{k} {\bR}^{\sf e}_{m} \bff_{k}}{\sum^{K}_{k = 1} \gamma^{\sf e}_{m}{\bff}^{\sf H}_{\ell} {\bR}^{\sf e}_{m} {\bff}_{\ell} + {\sigma}^{2}_{\sf e}/P}} \log_2 e
     \\
     =&\; \sqrt{\bar{\CMcal{V}}^{\sf e}_{m,k}},
\end{align}
where $(b)$ and $(c)$ follow from Jensen's inequality,
and  $(d)$ comes from the first-order  Taylor series approximation based on statistical linearization \cite{dalir2019maximizing}. 
\end{proof}

Applying \eqref{eq:UB_Imperfect_2} to \eqref{eq:UB_Imperfect}, we have the lower bound of secrecy rate approximated as
\begin{align}
     {R}^{\sf sec, p}_k \gtrapprox&\; R_k - \sqrt{\frac{\CMcal{V}_{k}}{L}}Q^{-1}(\epsilon_k) - \max_{m \in \CMcal{M}}\left\{\bar{R}^{\sf e}_{m,k} + \sqrt{\frac{\bar{\CMcal{V}}^{\sf e}_{m,k}}{L}}Q^{-1}\left(\delta_{m,k}\right)\right\}
     \\
     \label{eq:SecrecyRate_Imperfect}
     =&\; \bar{R}^{\sf sec,p}_k. 
\end{align}
Replacing ${R}_k^{\sf sec}$ in \eqref{eq:obj_func} with $\bar{R}_k^{\sf sec, p}$, we can also develop the joint secure precoding algorithm with minimizing the maximum error probability and information leakage rate.
In particular, as 
in Section \ref{subsec:Phase1}, we formulate the sum secrecy rate maximization problem with $\bar{R}_k^{\sf sec, p}$ as
\begin{align} 
    \label{eq:obj_func_imperfect2}
    &\mathop{{\text{maximize}}}_{\bF} \;\; \sum_{k = 1}^{K}\bar{R}_{k}^{\sf sec,p}(\bF),
    \\
    &{\text{subject to}} \;\; {\rm {tr}}\left(\bF \bF^{\sf H}\right) \leq 1.
\end{align}
We note that the formulated problem in \eqref{eq:obj_func_imperfect2} consists of the eavesdropper's  channel covariance matrix $\bR_m^{\sf e}$.
As in Section \ref{subsec:Phase1}, we also derive the first-order KKT condition for \eqref{eq:obj_func_imperfect2} as 
\begin{align}
    \label{eq:NEP_bar}
    \bar{\mathbf{B}}^{-1}_{\sf KKT}(\bar{\mathbf{f}})\bar{\mathbf{A}}_{\sf KKT}(\bar{\mathbf{f}})\bar{\mathbf{f}} = \bar{\lambda}(\bar{\mathbf{f}})\bar{\mathbf{f}},
\end{align}
where
\begin{align}
    \label{eq:A_KKT_bar}
    &\bar{\bA}_{\sf KKT}(\bar{\mathbf{f}}) = \sum_{k=1}^{K}\left[\frac{\omega_k}{\ln 2}\left(\frac{\mathbf{A}_k}{\bar{\mathbf{f}}^{\sf H}\mathbf{A}_{k}\bar{\mathbf{f}}}\right) + \frac{1}{\alpha}\left\{\sum_{m=1}^{M}\left(\frac{\omega^{\sf e}_{m,k} {\color{black} \beta_{m,k}} \left(\frac{\bar{\mathbf{f}}^{\sf H}\bar{\bC}_{m}\bar{\mathbf{f}}}{\bar{\mathbf{f}}^{\sf H}\bar{\bD}_{m,k}\bar{\mathbf{f}}}\right)^{\omega_{m,k}^{\sf e} }\frac{\bar{\bD}_{m,k}}{\bar{\mathbf{f}}^{\sf H}\bar{\bD}_{m,k}\bar{\mathbf{f}}}}{\sum_{\ell =1}^{M} {\color{black} \beta_{m,k}} \left(\frac{\bar{\mathbf{f}}^{\sf H}\bar{\bC}_{\ell}\bar{\mathbf{f}}}{\bar{\mathbf{f}}^{\sf H}\bar{\bD}_{\ell,k}\bar{\mathbf{f}}} \right)^{\omega_{\ell,k}^{\sf e}}}\right)\!\right\}\right]\!\bar{\lambda}_{\sf num}(\bar{\mathbf{f}}), 
    \\
    \label{eq:B_KKT_bar}
    &\bar{\bB}_{\sf KKT}(\bar{\mathbf{f}}) = \sum_{k=1}^{K}\left[\frac{\omega_k}{\ln 2}\left(\frac{\mathbf{B}_k}{\bar{\mathbf{f}}^{\sf H}\mathbf{B}_{k}\bar{\mathbf{f}}}\right) + \frac{1}{\alpha}\!\left\{\sum_{m=1}^{M}\left(\frac{\omega^{\sf e}_{m,k}{\color{black}\beta_{m,k}} \left(\frac{\bar{\mathbf{f}}^{\sf H}\bar{\bC}_{m}\bar{\mathbf{f}}}{\bar{\mathbf{f}}^{\sf H}\bar{\bD}_{m,k}\bar{\mathbf{f}}}\right)^{\omega_{m,k}^{\sf e} }\frac{\bar{\bC}_{m}}{\bar{\mathbf{f}}^{\sf H}\bar{\bC}_{m}\bar{\mathbf{f}}}}{\sum_{\ell=1}^{M}{\color{black} \beta_{m,k}}\left(\frac{\bar{\mathbf{f}}^{\sf H}\bar{\bC}_{\ell}\bar{\mathbf{f}}}{\bar{\mathbf{f}}^{\sf H}\bar{\bD}_{\ell,k}\bar{\mathbf{f}}} \right)^{\omega_{\ell,k}^{\sf e}}}\right)\!\right\}\right]\!\bar{\lambda}_{\sf den}(\bar{\mathbf{f}}),
    \\
    \label{eq:C_bar}
    &\bar{\mathbf{C}}_{m} =\;\mathrm{blkdiag}\left(\gamma^{\sf e}_{m} \mathbf{R}_{m}^{\sf e},\cdots,\gamma^{\sf e}_{m} \mathbf{R}_{m}^{\sf e}\right) + \mathbf{I}_{NK}\frac{\sigma^2_{\sf e}}{P} \in \mathbb{C}^{NK\times NK},
    \\
    \label{eq:D_bar}
    &\bar{\bD}_{m,k} =\;\bar{\mathbf{C}}_{m} - \mathrm{blkdiag}(0,\cdots,  {\gamma^{\sf e}_{m} \bR_{m}^{\sf e}},\cdots,0) \in \mathbb{C}^{NK\times NK},
    \\
    \label{eq:lamb_bar}
    &\bar{\lambda}(\bar{\mathbf{f}}) = \bar{\lambda}_{\sf num}(\bar{\mathbf{f}})/\bar{\lambda}_{\sf den}(\bar{\mathbf{f}}),
    \\
    &\bar{\lambda}_{\sf num}(\bar{\mathbf{f}}) = \prod_{k=1}^{K} \left(\frac{\bar{\mathbf{f}}^{\sf H}\mathbf{A}_{k}\bar{\mathbf{f}}}{\bar{\mathbf{f}}^{\sf H}\mathbf{B}_{k}\bar{\mathbf{f}}}\right)^{\omega_k},
    \; \bar{\lambda}_{\sf den}(\bar{\mathbf{f}}) = \prod_{k=1}^{K} \left\{\sum_{m=1}^{M}{\color{black} \beta_{m,k}} \left(\frac{\bar{\mathbf{f}}^{\sf H}{\bar{\bC}}_{m}\bar{\mathbf{f}}}{\bar{\mathbf{f}}^{\sf H}\bar{\bD}_{m,k}\bar{\mathbf{f}}}\right)^{\omega_{m,k}^{\sf e}}\right\}^{\frac{\ln 2}{\alpha }}.
\end{align}
Recall that $\mathbf{A}_k$ and $\mathbf{B}_k$ are derived in \eqref{eq:A_k} and \eqref{eq:B_k}, respectively.
Replacing $\mathbf{A}_{\sf KKT}$ and $\mathbf{B}_{\sf KKT}$ with $\bar{\mathbf{A}}_{\sf KKT}$ and $\bar{\mathbf{B}}_{\sf KKT}$ in Algorithm~\ref{alg:algorithm_1}, we can design the secure precoding algorithm by leveraging the partial CSIT of the eavesdroppers.
The maximum error probability and information leakage rate can also be optimized by following the same steps in Section \ref{subsec:Phase2} based on $\bar{R}^{\sf sec, \sf p}_k$ as
\begin{align}
    &\bar{\tau}^{\star} = Q\left(\sqrt{2 \ln \left(\frac{\sqrt{L}(1-w)\bar{R}_{\infty}}{\hat{\epsilon}_{\sf max}w\sqrt{2\pi}\sum_{k=\ell}^{K}\sqrt{{\CMcal{V}}_k}}\right)}\right),
    \\
    &\bar{\xi}^{\star} = Q\left(\sqrt{2 \ln \left(\frac{\sqrt{L}(1-w)\bar{R}_{\infty}}{\hat{\delta}_{\sf max}w\sqrt{2\pi}\sum_{m=j(k)}^{M}\sum_{k=1}^{K}\sqrt{\bar{\CMcal{V}}_{m,k}^{\sf e}}}\right)}\right),
\end{align}
where $\bar{R}_{\infty}$ is the normalization constant based on $\bar{R}^{\sf sec, \sf p}_k$ for $L= \infty$.
Then, the joint secure precoding algorithm is proposed with Algorithm~\ref{alg:algorithm_2} by replacing $\bA_k$, $\bB_k$, $\tau^{\star}$, and $\xi^{\star}$ with $\bar{\bA}_k$, $\bar{\bB}_k$, $\bar{\tau}^{\star}$, and $\bar{\xi}^{\star}$, respectively.

\section{Numerical Results}
In this section, we evaluate the performance of our proposed algorithms along with other benchmarks and 
deliver key insights.
Considered algorithms are:
(1) Algorithm 1, (2) Algorithm 1 with the channel covariance of wiretaps (Cov) in Section \ref{sec:imperfect_csi}, (3) Algorithm 2, (4) Algorithm 2 (Cov), (5) the FBL-based SE maximization algorithm (FBL-SE-MAX)  \cite{choi:iotj:21}, (6) regularized ZF (RZF), (7) RZF including the eavesdroppers (RZF-EVE), (8) ZF, (9) ZF-EVE, and (10) MRT.
Specifically, RZF-EVE and ZF-EVE algorithms exploit the wiretap channels
by constructing an effective channel matrix with the $N-K$ strongest wiretap channels.

To generate the channel vectors $\bh_k$ and $\bg_m$, we adopt a one-ring model \cite{adhikary2013joint} based on its spatial covariance matrices $\bR_k = \bbE[\bh_k \bh_{k}^{\sf H}]$ and $\bR^{\sf e}_{m} = \bbE[\bg_m \bg_m^{\sf H}]$.
In addition, we consider that the geometric location of each eavesdropper is correlated with a particular legitimate user. 
Accordingly, the channel angle of departure (AoD) of eavesdropper $m$ follows $\theta^{\sf e}_m \sim \theta_k + \cU (-\Delta \pi, \Delta \pi)$ for randomly selected user $k$ with a scalar weight $0< \Delta < 1 $, where $\theta_k $ represents the  AoD of user $k$.
We generate $\hat{\epsilon}_k$ and $\hat{\delta}_{m,k}$ uniformly spaced from $10^{-6}$ to $2 \times 10^{-6}$ for all users and eavesdroppers.
For simplicity, we consider $\hat{\delta}_{m,k} = \hat{\delta}_{m,k'}, \forall m \in \CMcal{M}$.
We set the scalar weight, coding length, optimization weight, iteration thresholds, and maximum iteration counts as $\Delta = 0.1$, $L=200$, $w=0.01$, $\varepsilon = \varepsilon^{\rm out} = 0.01$, and $t_{\rm max} = 15$ and $t_{\rm max}^{\rm out} = 5$ unless mentioned otherwise.

In addition, to generate the large-scale channel fading terms $\gamma_k$ and $\gamma_{m,k}^{\sf e}$, we adopt the ITU-R indoor pathloss model \cite{PATH:ITU:2013} which considers a non-line-of-sight pathloss environment by setting bandwidth, carrier frequency, distance power-loss coefficient, and noise figure to be 10 MHz, 5.2 GHz, 31 (equivalent to the pathloss exponent of $3.1$), and 5 dB, respectively.
We assume the noise power spectral density of legitimate users and
eavesdroppers is the same as $-174$ dBm/Hz. 
The users are randomly generated around the AP with the maximum distance of $50\;m$ and minimum distance of $5\;m$ from the AP.
The eavesdroppers are randomly distributed around random users with the maximum distance of $5\;m$ from the users.

\subsection{Secure Precoding for Given Error and Leakage in FBL Regime}
\begin{figure}[!t]
    \centering
    $\begin{array}{c c}
    {\resizebox{0.48\columnwidth}{!}
    {\includegraphics{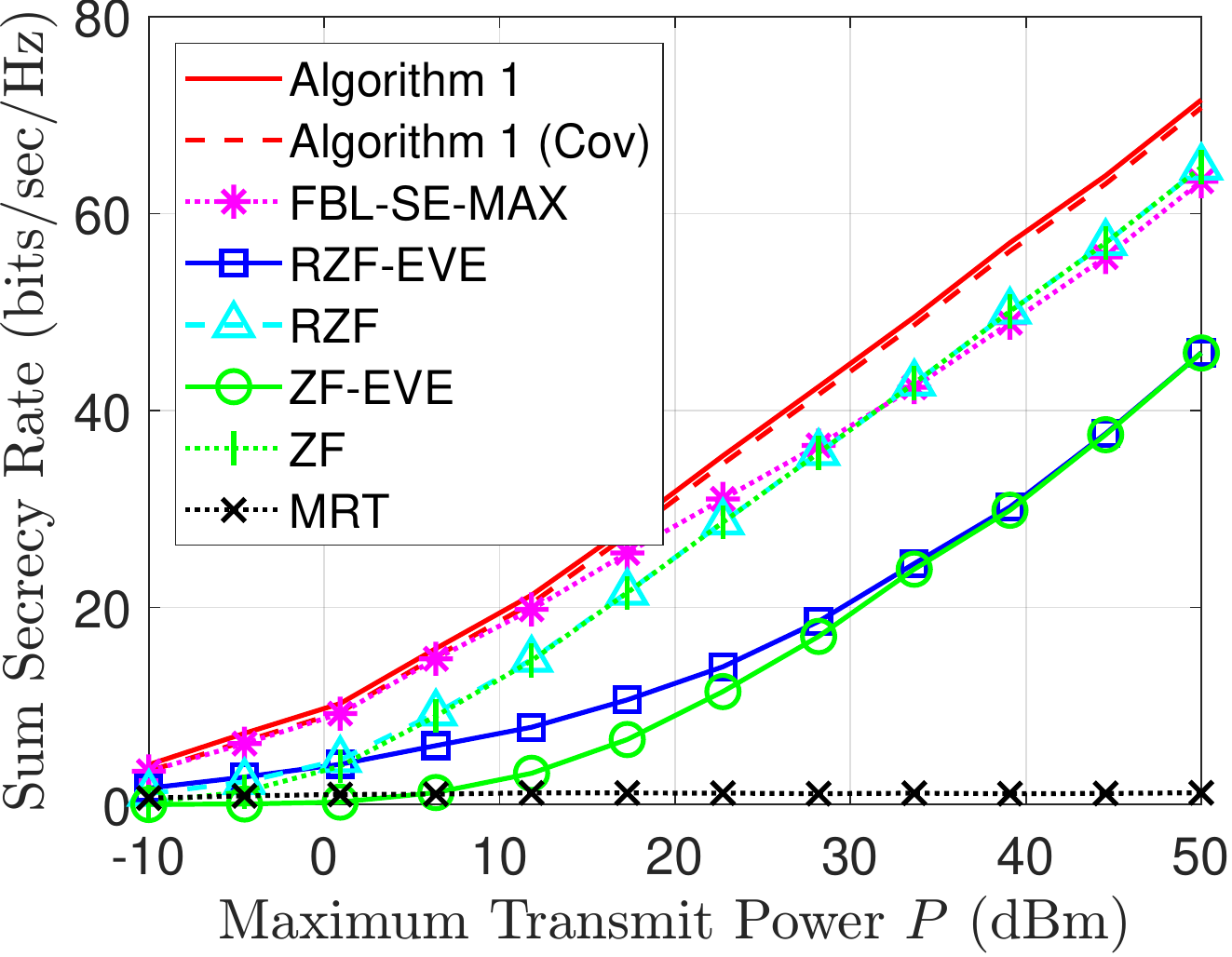}}
    }&
    {\resizebox{0.48\columnwidth}{!}
    {\includegraphics{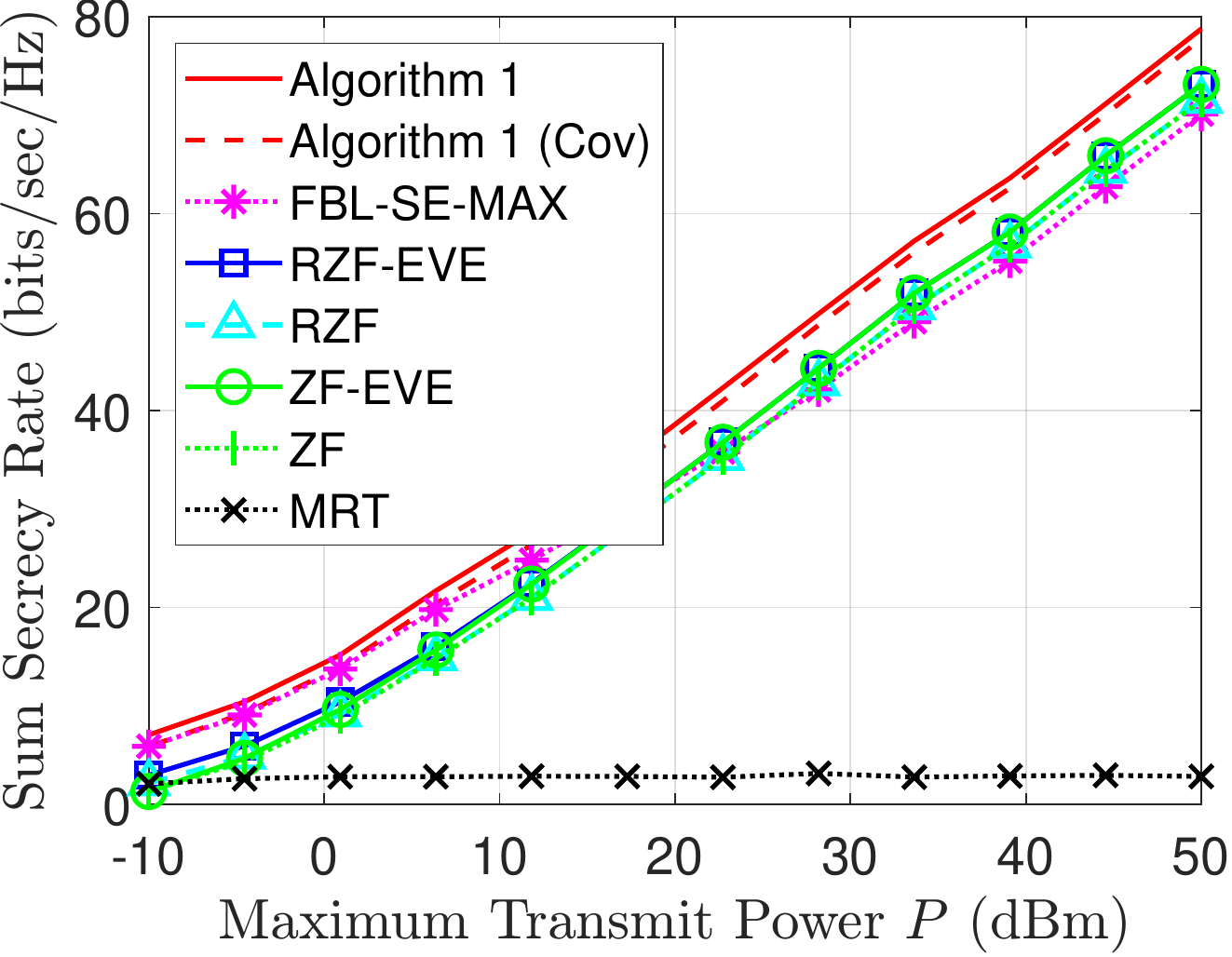}}
    }\\ \mbox{\small (a) $N = 8$ } & \mbox{\small (b) $N = 16$ }
    \end{array}$
      \vspace{-0.5em}
    \caption{The sum secrecy rate vs. maximum transmit power with (a) $N = 8$ and (b) $N = 16$ AP antennas for $K = 4$ users, and $M = 4$ eavesdroppers.}
    \label{fig:RatevsSNR}
      \vspace{-1em}
\end{figure}
We first investigate the precoders that maximize the sum secrecy rate for given $\epsilon_k$ and $\delta_{m,k}$ as $\hat{\epsilon}_k$ and $\hat{\delta}_{m,k}$, respectively.
In Fig.~\ref{fig:RatevsSNR}, we evaluate the sum secrecy rate of the simulated algorithms with respect to the transmit power $P$ for (a) $N=8$ and (b) $N=16$ with  $K=4$ and $M=4$.
Fig.~\ref{fig:RatevsSNR} shows that the proposed algorithms achieve the highest sum secrecy rate  compared to the baseline algorithms
because the proposed algorithms consider the effect of back-off factors due to the FBL as well as the wiretap channels.
Algorithm~\ref{alg:algorithm_1} (Cov), in particular, provides higher performance than the benchmarks without the perfect wiretap CSIT.
We observe that the secure precoders such as the RZF-EVE and ZF-EVE waste their extra dimension in nullifying the leakage from the wiretap channels rather than increasing legitimate users' channel gains, showing poor performance in Fig.~\ref{fig:RatevsSNR}(a).
However, as shown in Fig.~\ref{fig:RatevsSNR}(b), when the AP has the enough antennas, the ZF-based precoders outperform the SE maximization precoder.

\begin{figure}[!t]    
    {\centerline{\resizebox{0.45\columnwidth}{!}{\includegraphics{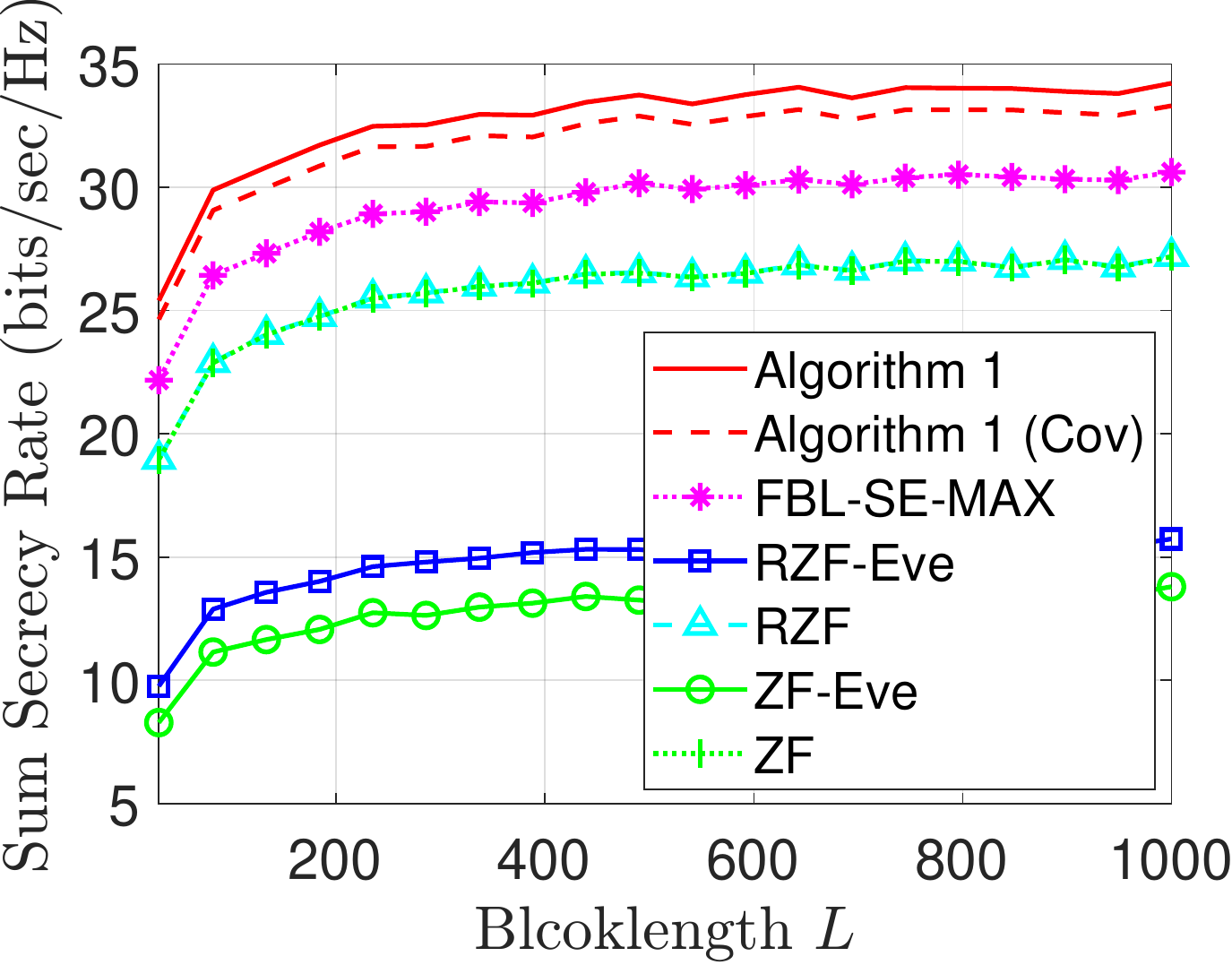}}}
      \vspace{-0.5em}
    \caption{The sum secrecy rate vs. the blocklength $L$ for $N=8$ AP antennas, $K=4$ users, $M=4$ eavesdroppers, and $P = 20$ dBm transmit power.}
    \vspace{-1em}
    \label{fig:RatevsL}}
\end{figure}
In Fig.~\ref{fig:RatevsL}, we evaluate the secrecy rate performance versus the
blocklength $L$ for $N=8$, $K=4$, $M=4$, and $P = 20$ dBm.
We note that all algorithms reveal an increasing trend of the sum secrecy rate with $L$ as theoretically shown in \eqref{eq:SecrecyRate}.
Fig.~\ref{fig:RatevsL} shows that the proposed algorithms maintain the highest secrecy rate performance regardless of $L$.
Accordingly, the proposed algorithm is effective in the general blocklength.
Thus, the proposed secure precoder is preferable in supporting the conventional as well as FBL-based communication systems.

\begin{figure}
    \centering
    $\begin{array}{c c}
    {\resizebox{0.48\columnwidth}{!}
    {\includegraphics{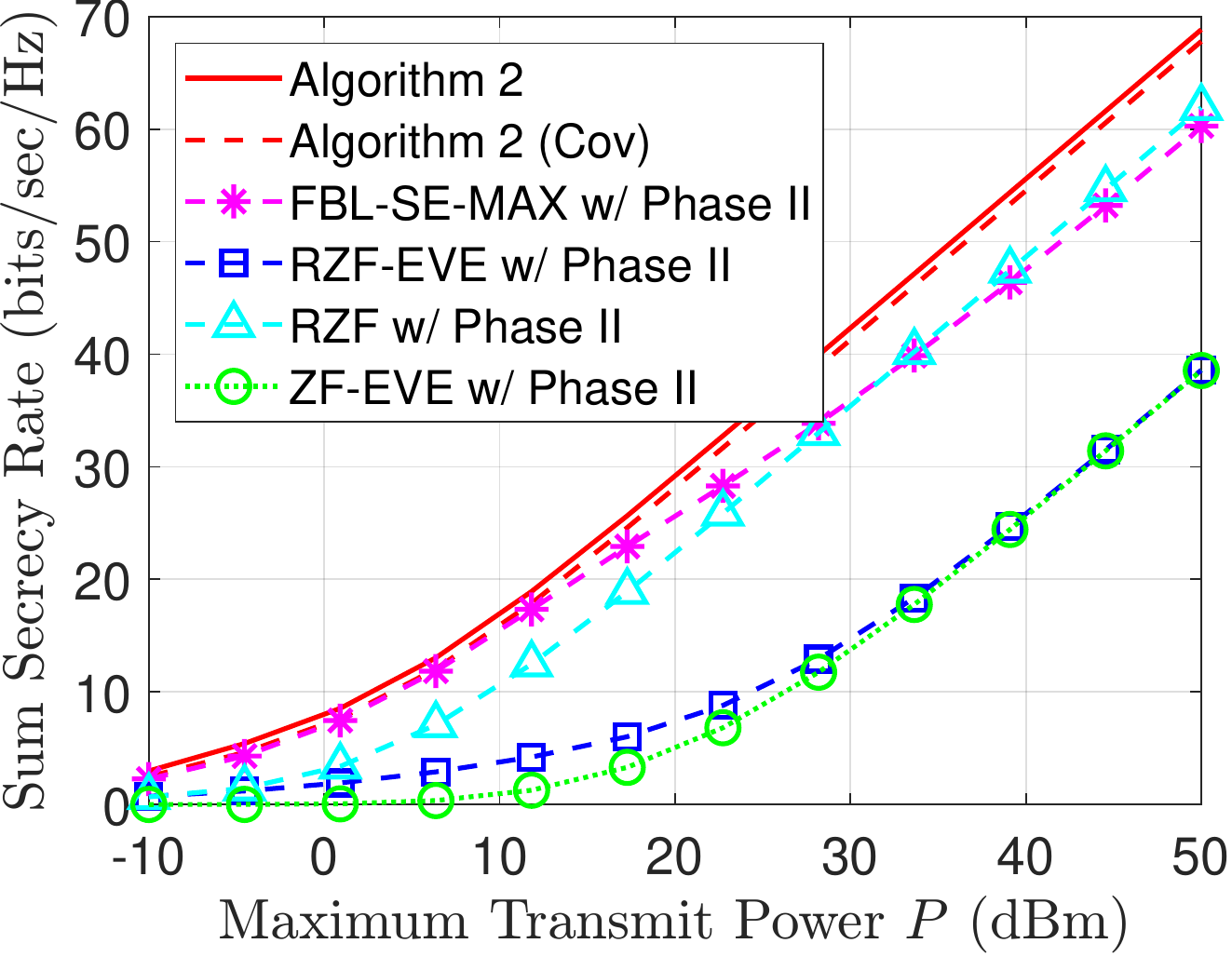}}
    }
    &
    {\resizebox{0.52\columnwidth}{!}
    {\includegraphics{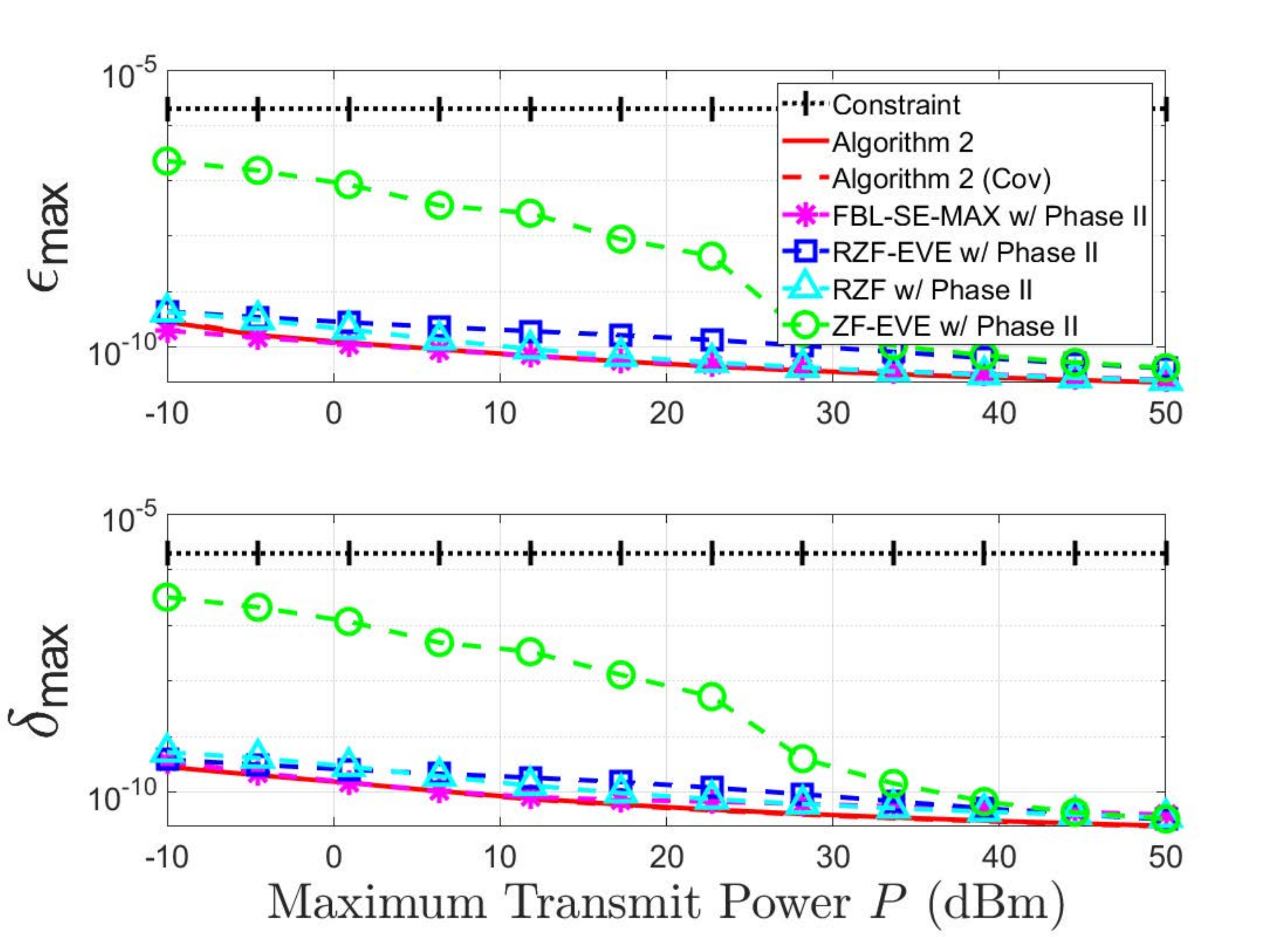}}
    }\\ \mbox{\small (a) Sum secrecy rate} & \mbox{\small \!(b) \!Maximum \!error \!probability \!and \!information \!leakage}
    \end{array}$
      \vspace{-0.5em}
    \caption{(a) The sum secrecy rate and (b) maximum error probability and information leakage rate vs. the maximum transmit power $P$ for $N= 8$ AP antennas, $K=4$ users, and $M=8$ eavesdroppers.}
      \vspace{-1em}
    \label{fig:ErrorMin}
\end{figure}
\subsection{Joint Secure Precoding with Error and Leakage Minimization in FBL Regime}
Now, we evaluate the proposed method in Algorithm~\ref{alg:algorithm_2}.
For comparison, the baseline algorithms utilize the second  phase  (Phase II in Section~\ref{subsec:Phase2}) of the proposed optimization framework for minimizing the maximum error probability and information leakage rate in Section \ref{subsec:Phase2}.
Fig.~\ref{fig:ErrorMin} shows that the performance of the sum secrecy rate, the maximum error probability, and information leakage rate with respect to the transmit power $P$ for $N=8$, $K=4$, and $M=8$.
The benchmark algorithms with the proposed minimization method provide the significant reduction in the maximum error probability and information leakage rate.
We note that the proposed algorithms achieve the highest secrecy rate performance while maintaining the lowest maximum error probability and information leakage rate for the most signal-to-noise ratio (SNR).
FBL-SE-MAX with Phase II attains the lowest maximum error probability in most SNR regimes.
The security rate of FBL-SE-MAX, however, reveals an increasing gap from the proposed algorithm due to ignoring the wiretap channels.

\begin{figure}
    \centering
    $\begin{array}{c c}
    {\resizebox{0.46\columnwidth}{!}
    {\includegraphics{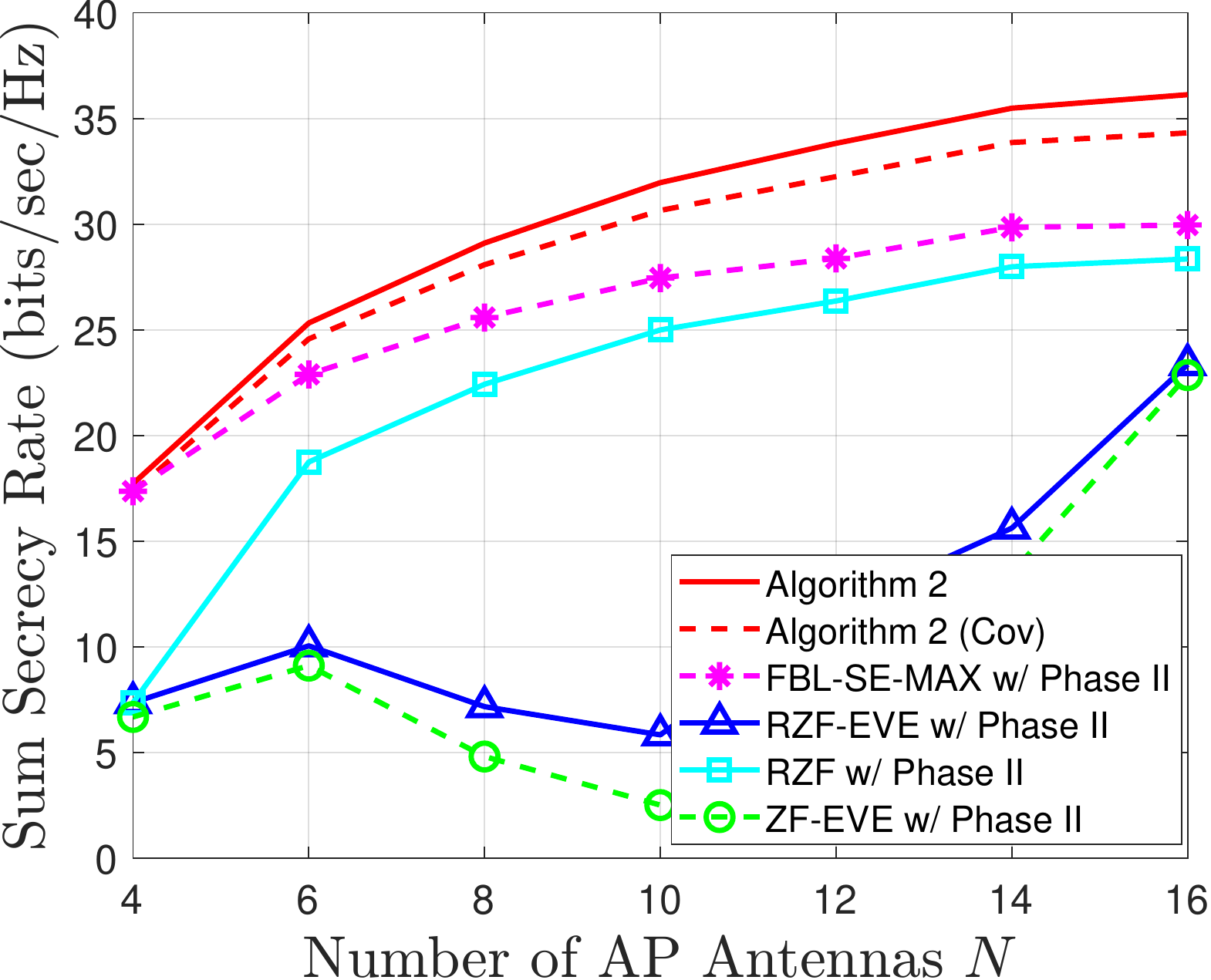}}
    }
    &
    {\resizebox{0.52\columnwidth}{!}
    {\includegraphics{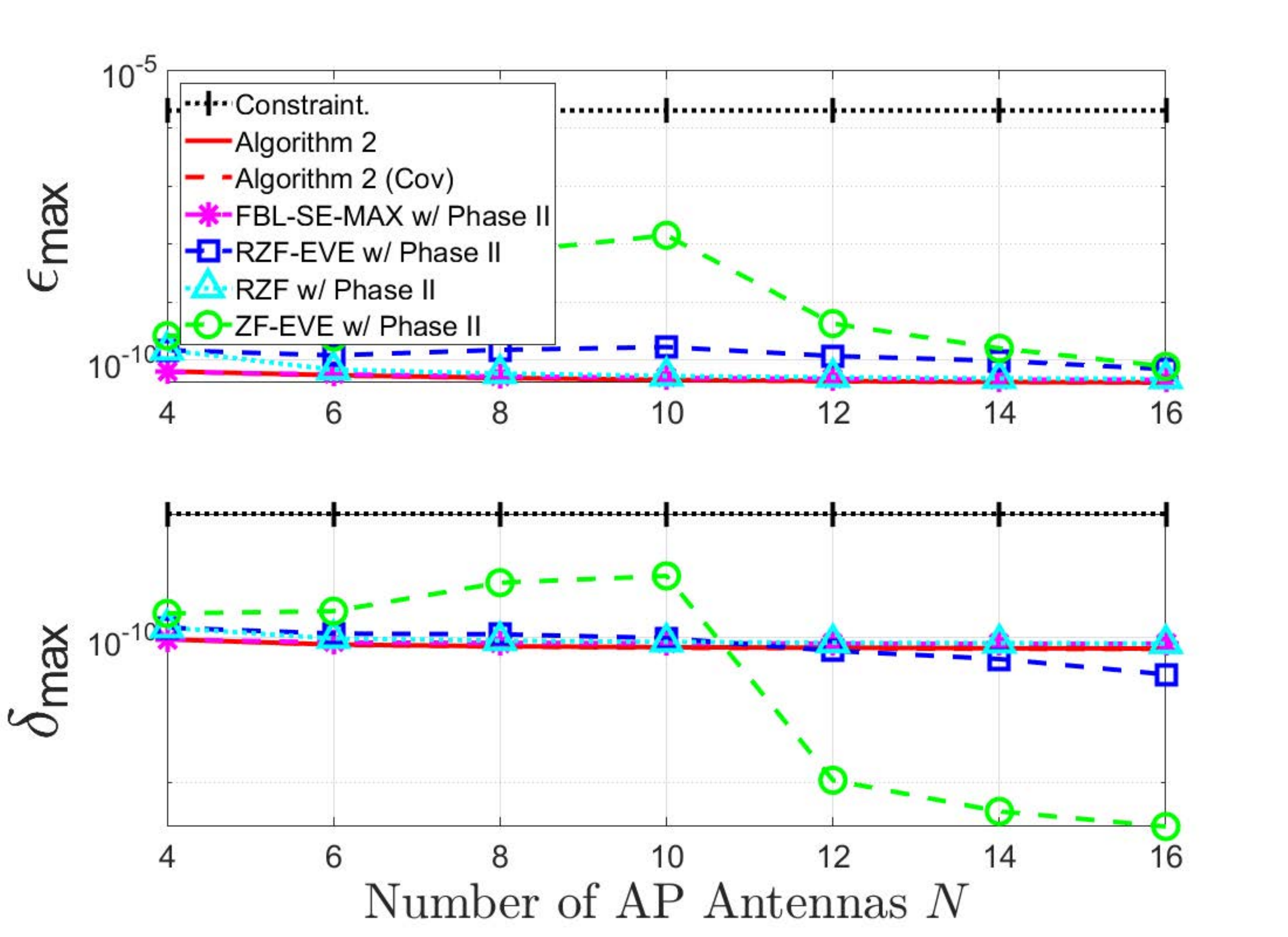}}
    }\\ \mbox{\small (a) Sum secrecy rate} & \mbox{\small \!(b) \!Maximum \!error \!probability \!and \!information \!leakage}
    \end{array}$
      \vspace{-0.5 em}
    \caption{(a) The sum secrecy rate and (b) maximum error probability and information leakage vs. the number AP antennas $N$ for $P  = 20$ dBm transmit power, $K=4$ users, and $M=8$ eavesdroppers.}
    \label{fig:ErrorMin_N}
      \vspace{-1em}
\end{figure}
To provide the numerical analysis in terms of the number of AP antennas $N$, we show the performance of the sum secrecy rate, the maximum error probability, and information leakage rate with respect to the number of AP antennas for $P = 20$ dBm, $K=4$, and $M=8$ in Fig.~\ref{fig:ErrorMin_N}.
It is shown that the proposed algorithms also achieve the highest secrecy rate performance regardless of  $N$.
As observed in Fig.~\ref{fig:ErrorMin_N}(a),  RZF-EVE and ZF-EVE reveal an unique trend of the sum secrecy rate as  $N$ increases.
Since there exists the optimal tradeoff between increasing the legitimate channel gain and nullifying the wiretap channels, the rates do not monotonically increase with $N$ for RZF-EVE and ZF-EVE.
For instance, when $N$ is not large enough, increasing channel gain should be more suitable to achieve the higher secrecy rate  than nullifying the wiretap channels with extra degree-of-freedom (DoF).
Thus,  RZF-EVE and ZF-EVE undergo poor performance in terms of the secrecy rate for small $N$, and they  enhance the secrecy rate with lager $N$.
Regarding the information leakage rate,  ZF-EVE  that  utilizes its extra DoF to fully nullify the wiretap channels is a near optimal solution with the enough DoF as shown in Fig.~\ref{fig:ErrorMin_N}(b).
Considering the entire performance, however, the proposed algorithms show the most balanced performance.

\begin{figure}
    \centering
    $\begin{array}{c c}
    {\resizebox{0.48\columnwidth}{!}
    {\includegraphics{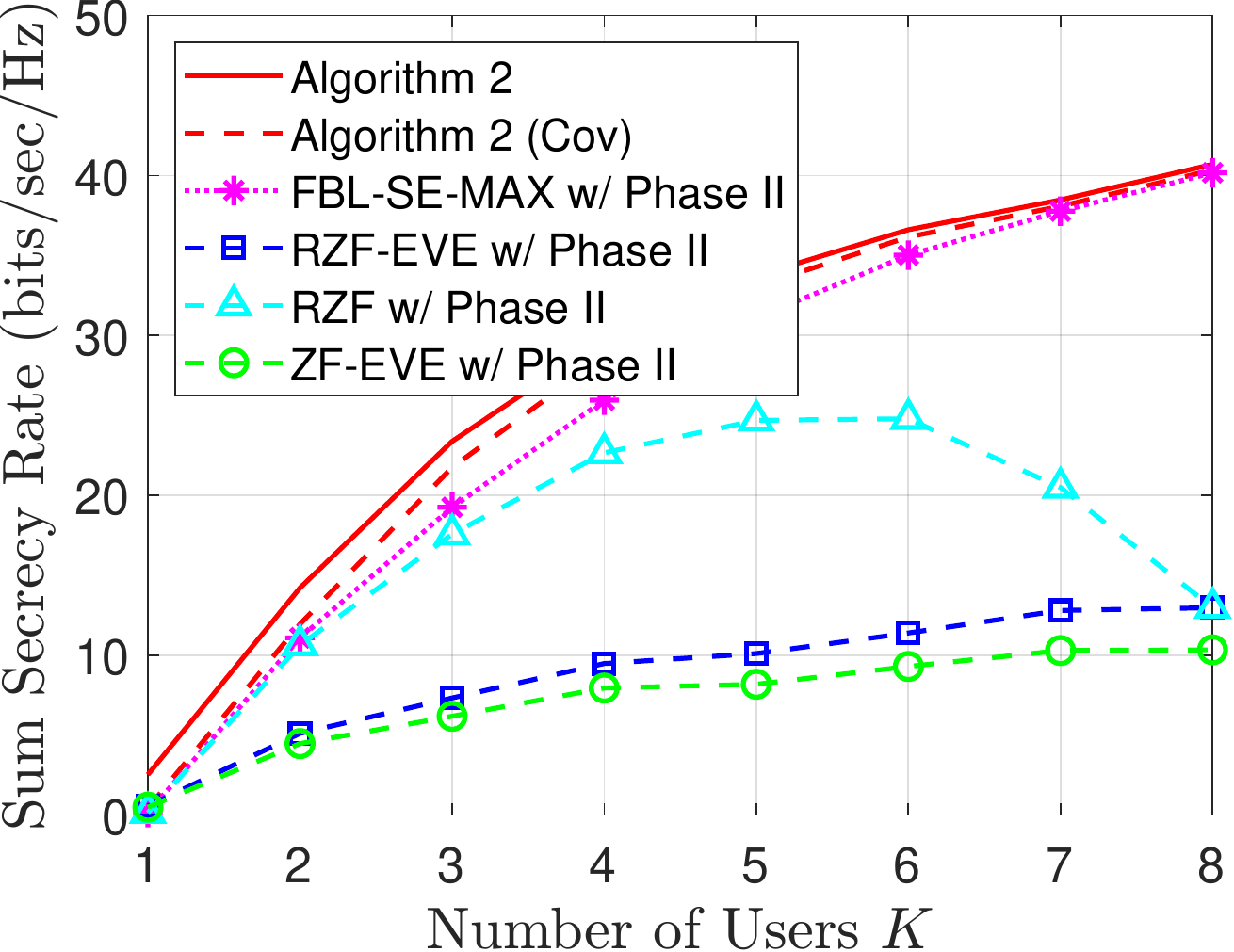}}
    }
    &
    {\resizebox{0.52\columnwidth}{!}
    {\includegraphics{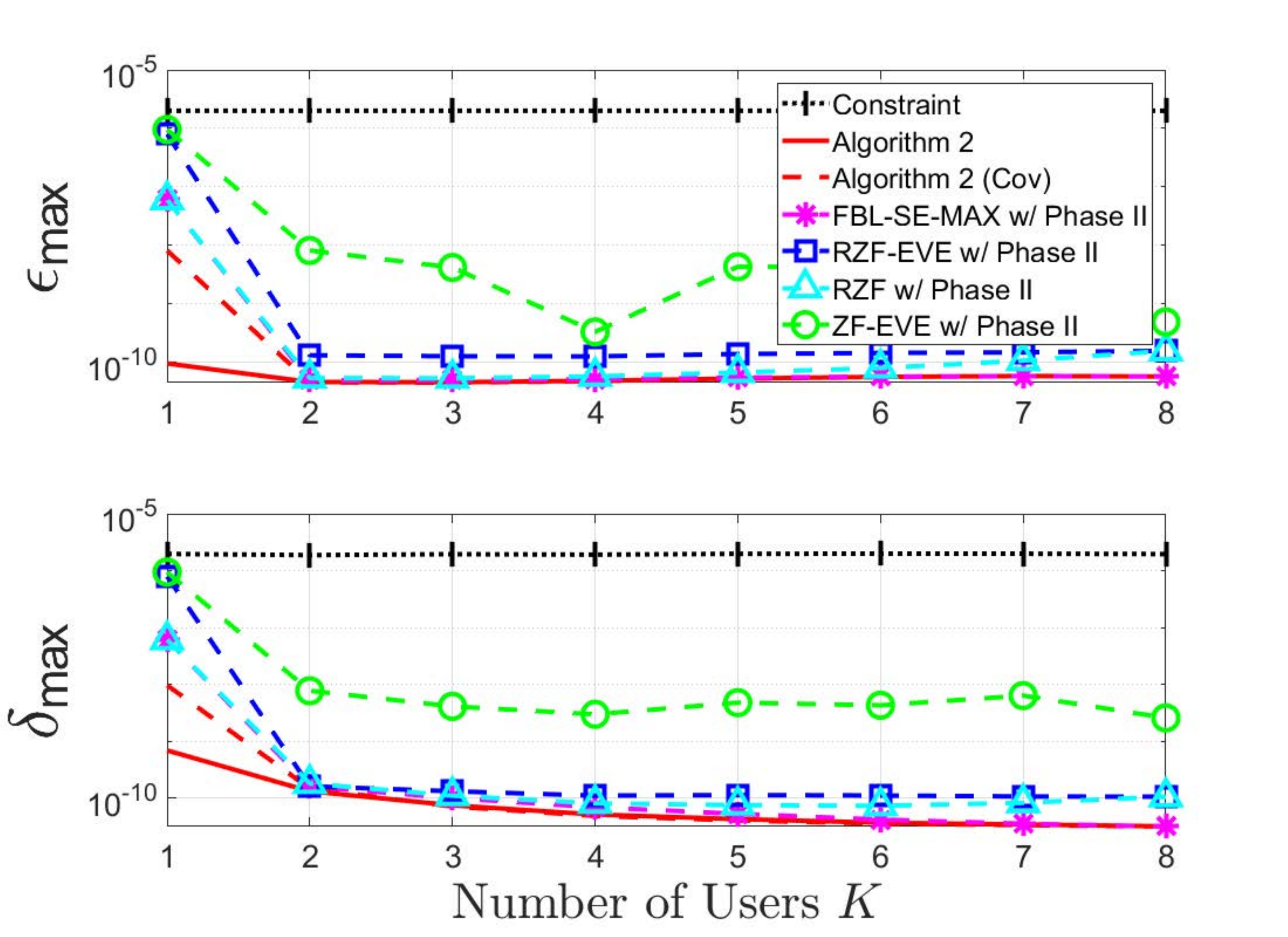}}
    }\\ \mbox{\small (a) Sum secrecy rate} & \mbox{\small \!(b) \!Maximum \!error \!probability \!and \!information \!leakage}
    \end{array}$
      \vspace{-0.5em}
    \caption{(a) The sum secrecy rate  and (b) maximum error and information leakage vs. the number of users $K$ for $P = 20$ dBm transmit power, $N= 8$ AP antennas, and $M=8$ eavesdroppers.}
    \label{fig:RatevsK}
      \vspace{-1em}
\end{figure}
We further evaluate the proposed algorithms with respect to the number of legitimate users $K$ for $N=8$, $M=8$, and $P=20$ dBm.
Overall, Fig.~\ref{fig:RatevsK} shows that the algorithms reveal similar ordering in the sum secrecy rate as in Fig.~\ref{fig:ErrorMin}.
However, the gaps between the proposed algorithms and FBL-SE-MAX becomes narrower as $K$ increases.
We conjecture that when DoF margin is small compared to the number of users, focusing on maximizing the sum rate provides a better tradeoff than allocating resources on mitigating information leakage.
Fig.~\ref{fig:RatevsK}(a) also reveals that there exists the optimal number of users for RZF in terms of the secrecy rate,
which is not the case for the proposed algorithms.
In Fig.~\ref{fig:RatevsK}(b), the proposed algorithm provides the lowest maximum error probability and information leakage rate regardless of $K$.
Therefore, the proposed algorithm provides improvement in rate, reliability, and security  from the conventional precoding algorithms for the considered system environment.

\begin{figure}
    \centering
    $\begin{array}{c c}
    {\resizebox{0.48\columnwidth}{!}
    {\includegraphics{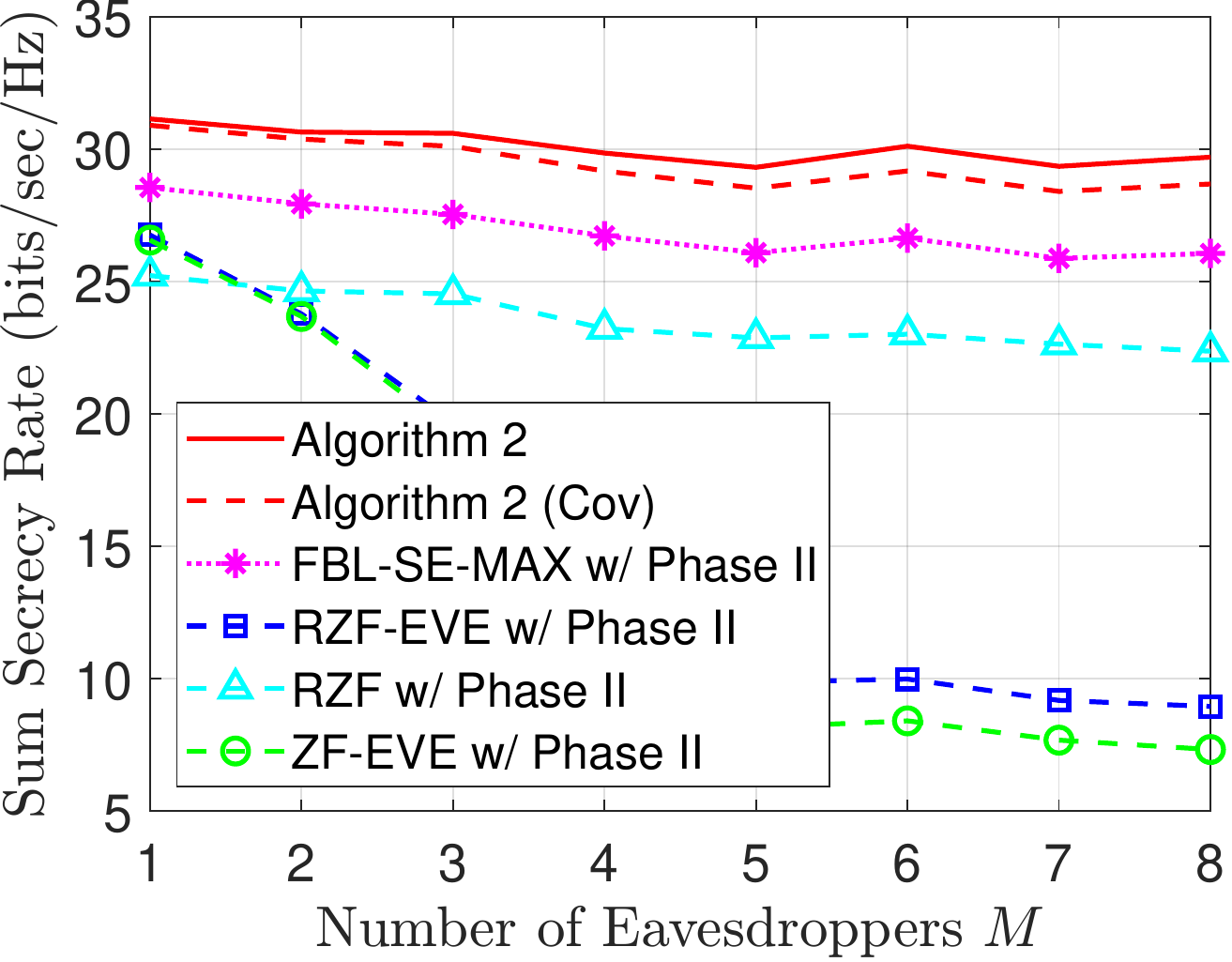}}
    }
    &
    {\resizebox{0.52\columnwidth}{!}
    {\includegraphics{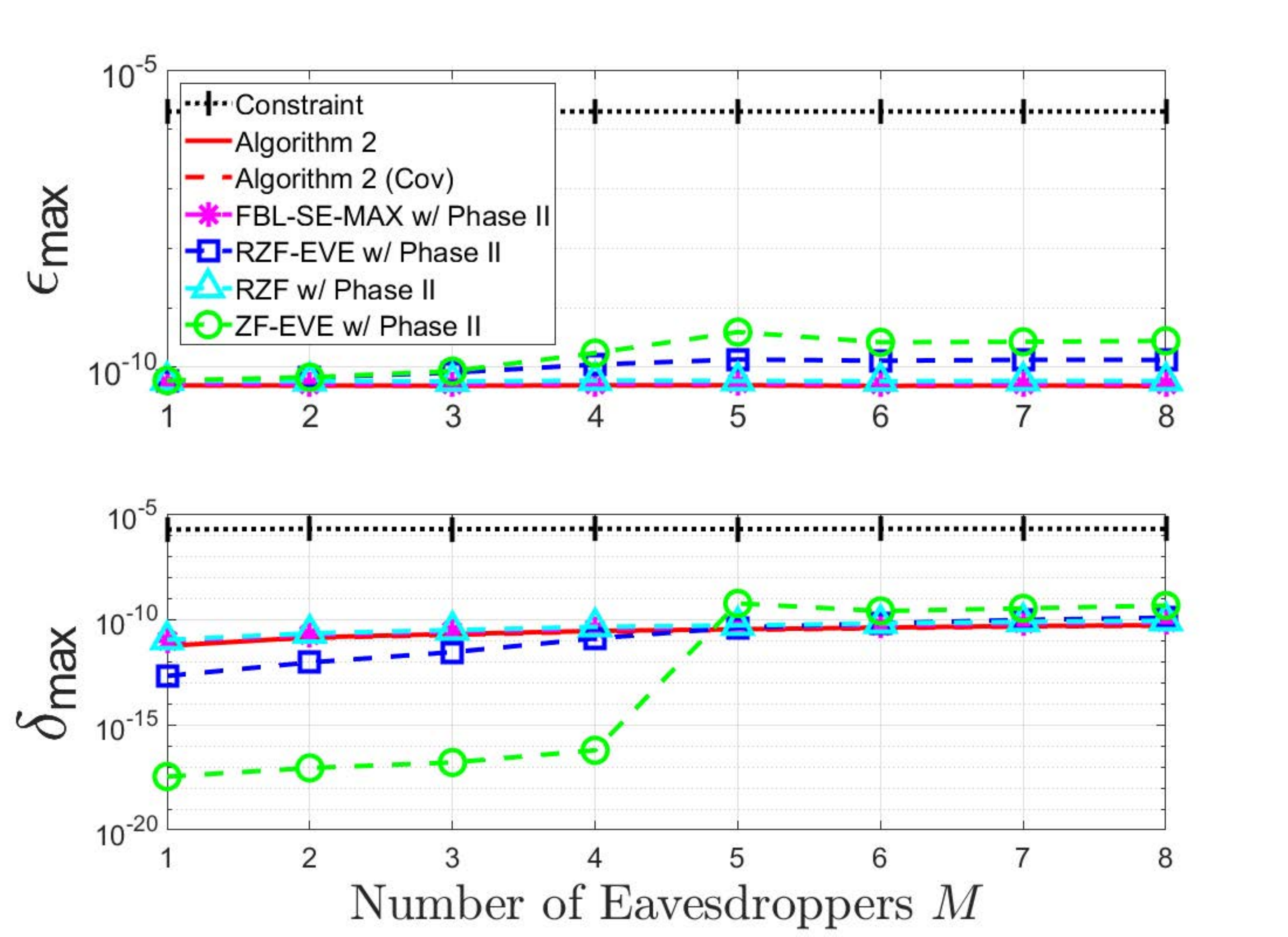}}
    }\\ \mbox{\small (a) Sum secrecy rate} & \mbox{\small \!(b) \!Maximum \!error \!probability \!and \!information \!leakage}
    \end{array}$
      \vspace{-0.5em}
    \caption{(a) The sum secrecy rate and (b) maximum error and information leakage vs. the number of eavesdroppers $M$ for $P = 20$ dBm transmit power, $N= 8$  antennas, and $K=4$ users.}
    \label{fig:RatevsM}
      \vspace{-1em}
\end{figure}
In Fig.~\ref{fig:RatevsM}, we assess the sum secrecy rate in terms of the number of eavesdroppers $M$ for $P=20$ dBm, $N=8$, and $K=4$.
Fig.~\ref{fig:RatevsM}(a) shows that the proposed algorithms still achieve the highest secrecy rates.
As $M$ increases, the relative gaps between the proposed and the other algorithms become larger, which demonstrate the effective secure precoding performance of the proposed algorithms.
We note that for RZF-EVE and ZF-EVE, decreasing the spatial DoF margin leads to the poor performance of the secrecy rate.
As shown in Fig.~\ref{fig:RatevsM}(b), the proposed algorithms offer good and robust performance in minimizing the maximum error probability and information leakage rate over $M$.
Although ZF-EVE shows the lowest maximum information leakage rate for $M \leq 4$ by fully nullifying the wiretap channels, the achieved secrecy rate is significantly lower and the leakage rate becomes higher than the proposed algorithms for $M>4$.
In this regard, the proposed algorithm is considered to be a potential PLS solution that is robust to the number of eavesdroppers.

\begin{figure}[!t]
    {\centerline{\resizebox{0.65\columnwidth}{!}{\includegraphics{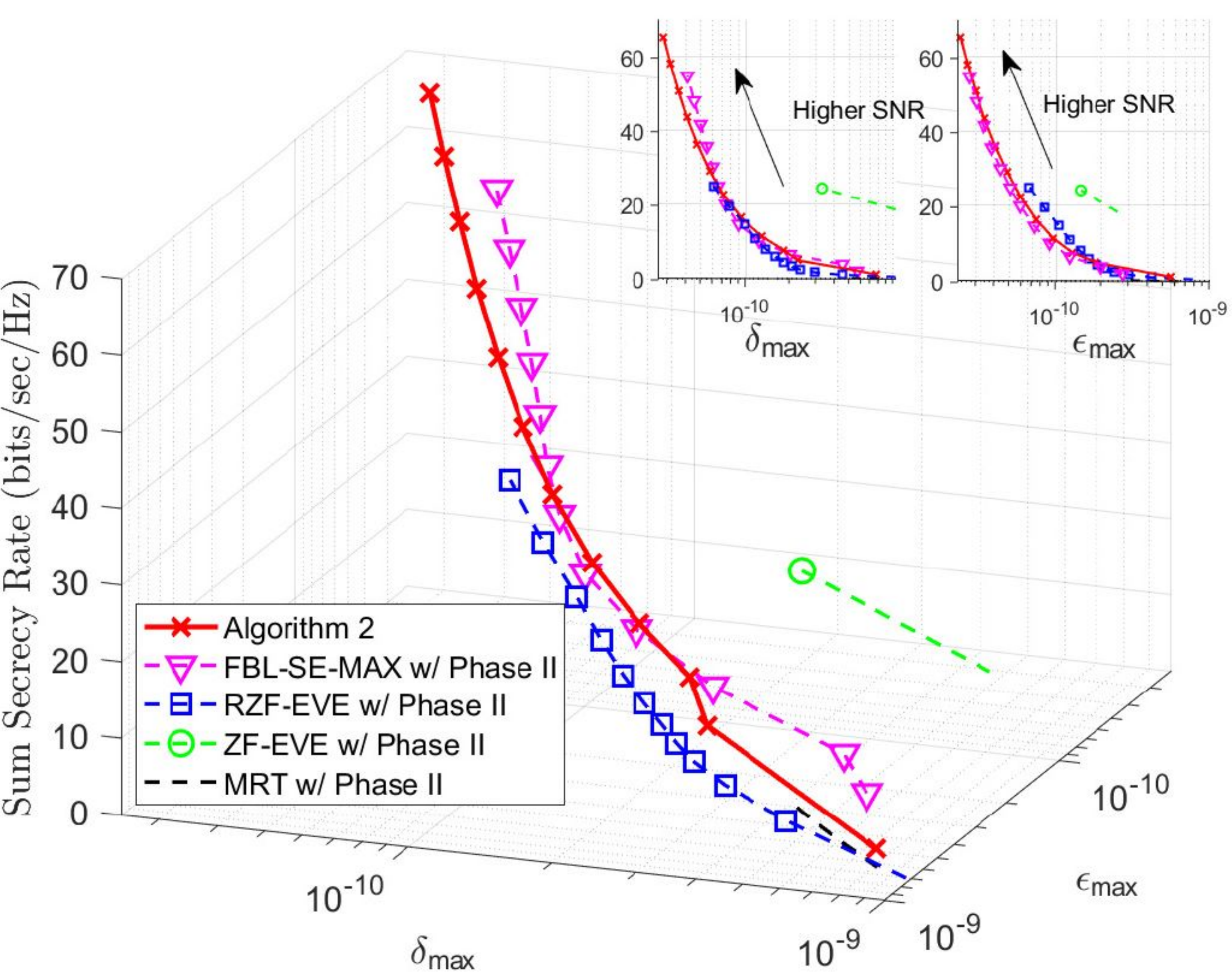}
    }}
      \vspace{-1em}
    \caption{The sum secrecy rate, maximum error probability, and maximum information leakage evaluated over transmit power $P$ from $0$ to $50$ dBm for $N=8$ AP antennas, $K=4$ users, and $M=8$ eavesdroppers.}
    \label{fig:3D_plot}}
      \vspace{-1em}
\end{figure}
For better visual understanding, we also plot the proposed algorithm by drawing the secrecy rate, maximum error probability, and maximum information leakage rate together over $P$ from $0$ to $50$ dBm for $N=8$, $K=4$, and $M=8$ in Fig.~\ref{fig:3D_plot}.
For instance, if the system requirements of the sum secrecy rate, maximum error probability, and maximum information leakage are $\sum_{k=1}^{K}R^{\sf sec}_k > 40$ bits/sec/Hz, $\max\{\epsilon_k\} < 4.0 \times 10^{-11}$, and $\max\{\delta_{m,k}\} < 4.0 \times 10^{-11}$, the proposed algorithm can satisfy the requirements with 33.3 dBm transmit power, whereas FBL-SE-MAX with Phase II requires 44.4 dBm, i.e., more than approximately $12\times$ the transmit power is needed, and the others cannot meet the requirements in the feasible transmit power regime.
As a result, the proposed algorithm provides the significant gain of SNR to fulfill the stringent requirements of URLLC in the FBL regime.

\begin{figure}[!t]    
    {\centerline{\resizebox{0.5\columnwidth}{!}{\includegraphics{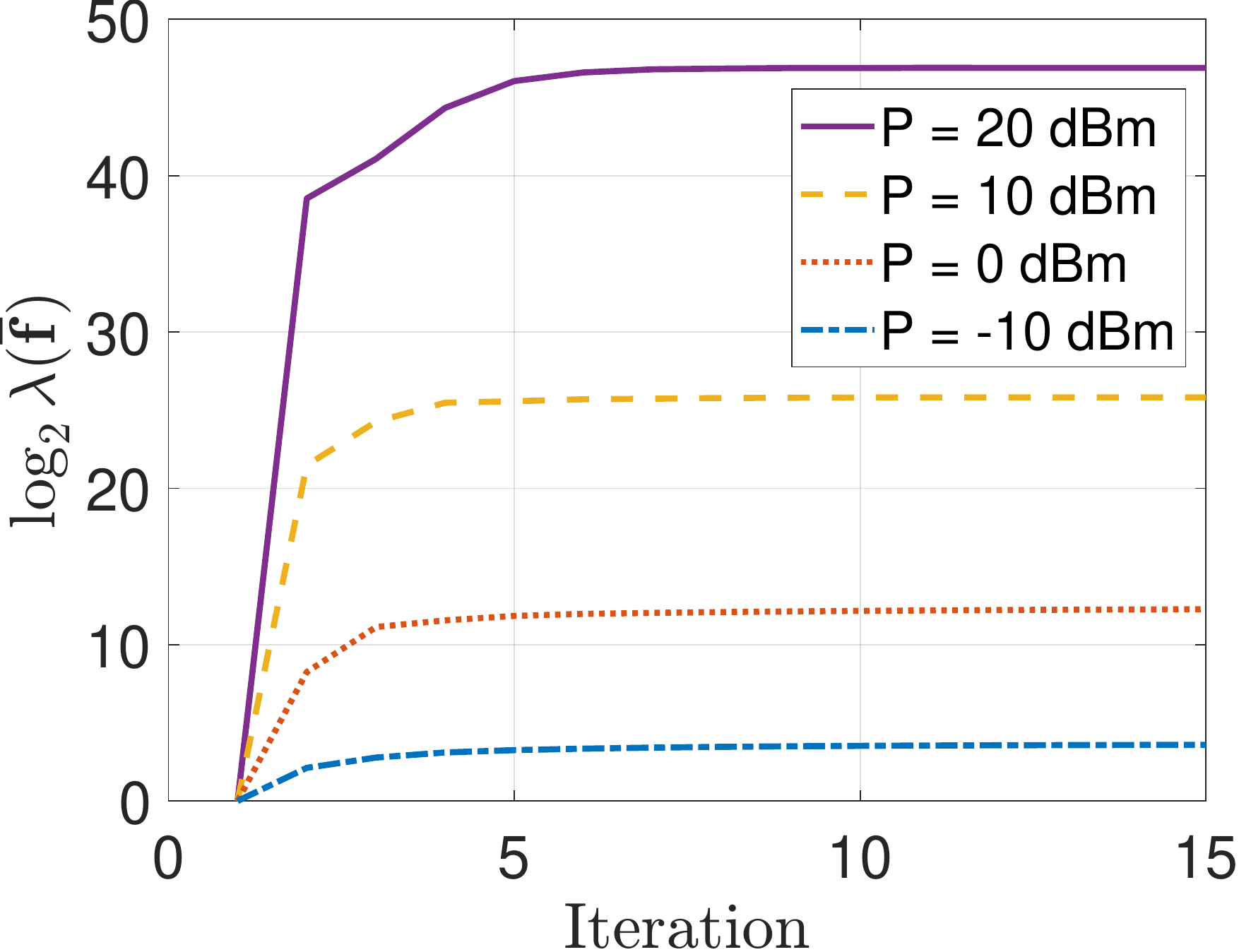}}}
      \vspace{-0.5em}
    \caption{Convergence results in terms of $\log_2 \lambda(\bar{\bf f})$ for $N=8$ AP antennas, $K=4$ users, $M=4$ eavesdroppers, and $P =\{-10, 0, 10, 20\}$ dBm transmit power.}
    \label{fig:Converg}}
      \vspace{-1em}
\end{figure}
In Fig.~\ref{fig:Converg}, we present the convergence results in terms of the approximated objective function $\log_2\lambda(\bar{\bf f})$ in \eqref{eq:lamb} for $P \in \{-10, 0, 10, 20\}$ dBm, $N=8$, $K=4$, and $M=4$.
As shown in Fig.~\ref{fig:Converg}, the proposed algorithm converges within $T=6$ iterations in the practical transmit power regime.
Accordingly, the proposed algorithm demonstrates high potential for the practical communications compared to other high-complexity precoding algorithms in the FBL regime.
Overall, the proposed algorithms offer significant improvement in the secrecy rate, error probability, and information leakage rate, and thus enhance the data rate, security,  and reliability with low complexity in the FBL regime.

\section{Conclusion}
In this paper, we proposed secure precoding algorithms that jointly optimize rate, reliability, and security in the finite blocklength regime.
To this end, we formulated a joint optimization problem for the sum secrecy rate and error probability, and information leakage rate according to the characteristics of finite blocklength channel coding-based communications.
To resolve the key challenges in solving the problem, we first decompose the problem into two sub-problems.
Then, for the first problem, we designed the precoding algorithm that maximizes the sum secrecy rate for the given error probability and information leakage rate by identifying stationary points.
For the second problem, we also reformulated the multi-objective optimization problem into the single-objective optimization problem for a given precoder and solved the KKT conditions to derive the optimal error probability and information leakage rate.
Through the alternating optimization, we developed the joint optimization algorithm that provides the significantly improved tradeoff among the security, the error probability, and information leakage rate.
We further extended the algorithms for the case in which only the partial CSI of wiretap channels is available.
Via simulations, we validated the secrecy rate, error probability, and information leakage performances of the proposed algorithms.
In particular, we demonstrated that the proposed methods achieve the highest secrecy performance while satisfying the stringent reliability and security requirements compared to the existing methods with fast convergence and high robustness.
Therefore, the proposed methods can play a key role in realizing the URLLC under high information security.

\appendices
\section{Proof of Lemma 1} \label{pf:KKT_condition}
\begin{proof}
The Lagrangian function of the problem \eqref{eq:Object_GPI} is 
\begin{align}
    \cL_1(\bar {\bf f}) = \sum_{k=1}^{K}\left[\log_2 \left(\frac{\bar{\mathbf{f}}^{\sf H}\mathbf{A}_{k}\bar{\mathbf{f}}}{\bar{\mathbf{f}}^{\sf H}\mathbf{B}_{k}\bar{\mathbf{f}}}\right)^{\omega_k} \!-\! \ln \left\{\sum_{m=1}^{M} \left(\frac{\bar{\mathbf{f}}^{\sf H}\mathbf{C}_{m}\bar{\mathbf{f}}}{\bar{\mathbf{f}}^{\sf H}\mathbf{D}_{m,k}\bar{\mathbf{f}}}\right)^{\omega_{m,k}^{\sf e}}\right\}^{ \frac{1}{\alpha}}\right].
\end{align}
According to the first-order KKT condition, the stationary points need to satisfy zero-gradient, i.e., $\frac{\partial \cL_1(\bar {\bf f})}{\partial \bar {\bf f}^{\sf H}} = 0$.
We take the partial derivative of $\cL_{1}(\bar {\bf f})$ with respect to ${\bf \bar f }$ and set it to zero. 
For simplicity, we denote the first and second part of the Lagrangian function as $\cL_{1, \sf user}(\bar {\bf f})$ and $\cL_{1, \sf eve}(\bar {\bf f})$, i.e., $\cL_{1}(\bar {\bf f}) = \cL_{1, \sf user}(\bar {\bf f}) - \cL_{1, \sf eve}(\bar {\bf f})$. 
Then, using 
\begin{align}
    \label{eq:derivative_matrix}
    \frac{\partial \left( \frac{{\bf \bar f }^{\sf H} {\bA}_k \bar {\bf f} }{{\bf \bar f }^{\sf H} {\bB}_k \bar {\bf f}} \right)}{\partial \bar {\bf f}^{\sf H}} = \left( \frac{{\bf \bar f }^{\sf H} {\bA}_k \bar {\bf f} }{{\bf \bar f }^{\sf H} {\bB}_k \bar {\bf f}} \right) \left[  \frac{{\bA}_k \bar {\bf f} }{{\bf \bar f }^{\sf H} {\bA}_k \bar {\bf f}} - \frac{{\bB}_k \bar {\bf f} }{{\bf \bar f }^{\sf H} {\bB}_k \bar {\bf f}} \right],
\end{align}
the partial derivative of $\cL_{1, \sf user}(\bar {\bf f})$ can be calculated as
\begin{align}
    \label{eq:partial_L1_user}
    \frac{\partial \cL_{1, \sf user} (\bar{\bf f})}{\partial \bar {\bf f}^{\sf H}} = \sum_{k=1}^K \frac{1}{\ln 2} \left( \frac{{\bA}_k \bar {\bf f} }{{\bf \bar f }^{\sf H} {\bA}_k \bar {\bf f}} - \frac{{\bB}_k \bar {\bf f} }{{\bf \bar f }^{\sf H} {\bB}_k \bar {\bf f}} \right).
\end{align}
We also take the partial derivative of $\cL_{1, \sf eve}(\bar {\bf f})$ with respect to ${\bf \bar f }$ as
\begin{align}
    \label{eq:partial_L1_eve}
    \frac{\partial \cL_{1, \sf eve} (\bar{\bf f})}{\partial \bar {\bf f}^{\sf H}} \!=\! \frac{1}{\alpha}\sum_{k=1}^K\sum_{m=1}^{M} \left\{\frac{\omega_{m,k}^{\sf e} \left( \frac{\bar{\bf f}^{\sf H} {\bC}_{m} \bar {\bf f} }{\bar{\bf f}^{\sf H} {\bD}_{m,k} \bar {\bf f} } \right)^{\omega_{m,k}^{\sf e}} \left(\frac{{\bC}_{m} \bar {\bf f} }{{\bf \bar f }^{\sf H} {\bC}_{m} \bar {\bf f}} - \frac{{\bD}_{m,k} \bar {\bf f} }{{\bf \bar f }^{\sf H} {\bD}_{m,k} \bar {\bf f}} \right)}{\sum_{\ell=1}^{M} \left( \frac{\bar{\bf f}^{\sf H} {\bC}_{\ell} \bar {\bf f} }{\bar{\bf f}^{\sf H} {\bD}_{\ell,k} \bar {\bf f}}\right)^{\omega_{\ell,k}^{\sf e}}}\right\}.
\end{align}
Then, using \eqref{eq:partial_L1_user} and \eqref{eq:partial_L1_eve}, the first-order KKT condition holds if
\begin{align}
    \label{eq:kkt_cond_rearrange}
    & \sum_{k=1}^{K}  \left[\frac{\omega_k}{\ln 2} \frac{\bA_k}{ {\bf \bar f }^{\sf H} \bA_k \bar {\bf f}} +  \frac{1}{\alpha}\left\{\sum_{m=1}^{M}\left(\frac{\omega^{\sf e}_{m,k} \left(\frac{\bar{\mathbf{f}}^{\sf H}\mathbf{C}_{m}\bar{\mathbf{f}}}{\bar{\mathbf{f}}^{\sf H}\mathbf{D}_{m,k}\bar{\mathbf{f}}}\right)^{\omega_{m,k}^{\sf e} }\frac{\mathbf{D}_{m,k}}{\bar{\mathbf{f}}^{\sf H}\mathbf{D}_{m,k}\bar{\mathbf{f}}}}{\sum_{\ell =1}^{M}\;\left(\frac{\bar{\mathbf{f}}^{\sf H}\mathbf{C}_{\ell}\bar{\mathbf{f}}}{\bar{\mathbf{f}}^{\sf H}\mathbf{D}_{\ell,k}\bar{\mathbf{f}}} \right)^{\omega_{\ell,k}^{\sf e}}}\right)\right\} \right]\bar {\bf f}
    \nonumber \\ 
    &= \sum_{k=1}^{K} \left[\frac{\omega_k}{\ln 2} \frac{\bB_k}{ {\bf \bar f }^{\sf H} {\bB}_k \bar {\bf f}} + \frac{1}{\alpha}\left\{\sum_{m=1}^{M}\left(\frac{\omega^{\sf e}_{m,k} \left(\frac{\bar{\mathbf{f}}^{\sf H}\mathbf{C}_{m}\bar{\mathbf{f}}}{\bar{\mathbf{f}}^{\sf H}\mathbf{D}_{m,k}\bar{\mathbf{f}}}\right)^{\omega_{m,k}^{\sf e} }\frac{\mathbf{C}_{m}}{\bar{\mathbf{f}}^{\sf H}\mathbf{C}_{m}\bar{\mathbf{f}}}}{\sum_{\ell=1}^{M}\;\left(\frac{\bar{\mathbf{f}}^{\sf H}\mathbf{C}_{\ell}\bar{\mathbf{f}}}{\bar{\mathbf{f}}^{\sf H}\mathbf{D}_{\ell,k}\bar{\mathbf{f}}} \right)^{\omega_{\ell,k}^{\sf e}}}\right)\right\} \right] \bar {\bf f}.
\end{align}
The first-order KKT condition is reorganized as
\begin{align}
    {\bf{A}}_{\sf KKT}(\bar {\bf f}) \bar {\bf f} = \lambda(\bar {\bf f}) {\bf{B}}_{\sf KKT} (\bar {\bf f}) \bar {\bf f}.
\end{align}
Since ${\bf{B}}_{\sf KKT} (\bar {\bf f}) $ is Hermitian, it is invertible. 
This completes the proof.
\end{proof}

\section{Proof of Lemma~\ref{lem:ErrorMin}} \label{pf:ErrorMin}
\begin{proof}
We note that $Q^{-1}(x)$ is convex if $x<\frac{1}{2}$, so that the problem in \eqref{eq:Error_obj} becomes a convex problem.
To solve the problem, we define the Lagrangian function of the problem in \eqref{eq:Error_obj} as
\begin{align}
    \cL_2 =&\; \frac{w}{R_{\infty}} \sum_{k=1}^{K} \left[\sqrt{\frac{\CMcal{V}_k}{L}}Q^{-1}(\epsilon_k)
    + \sum_{m=1}^{M}\left(R^{\sf e}_{m,k} + \sqrt{\frac{\CMcal{V}_{m,k}^{\sf e}}{L}}Q^{-1}\left(\delta_{m,k}\right)\right)\right]
    \nonumber \\ 
    &  + (1-w)\left(\frac{\tau}{\hat{\epsilon}_{\sf max}} \!+\! \frac{\xi}{\hat{\delta}_{\sf max}}\right) \!-\! \sum_{k=1}^{K}\lambda_k (\tau - \epsilon_k) - \sum_{k=1}^{K} \nu_k (\hat{\epsilon_k} - \epsilon_k) \!-\! \mu (\hat{\epsilon_{\sf max}} - \tau)
    \nonumber\\
    & -  \sum_{m=1}^{M}\sum_{k=1}^{K}\lambda_{m,k}^{\sf e} (\xi - \delta_{m,k}) - \sum_{m=1}^{M}\sum_{k=1}^{K} \nu_{m,k}^{\sf e} (\hat{\delta}_{m,k} - \delta_{m,k}) \!-\! \mu^{\sf e} (\hat{\delta}_{\sf max} - \xi),
\end{align}
where $\lambda_k$, $\nu_k$, $\mu$, $\lambda_{m,k}^{\sf e}$, $\nu_{m,k}^{\sf e}$, and $\mu^{\sf e}$ indicate Lagrangian multipliers.
For a feasible solution, we assume $\lambda_k \geq 0$, $\nu_k \geq 0, \forall k \in \CMcal{K}$, $\mu \geq 0$, $\mu^{\sf e} \geq 0$,  $\lambda_{m,k}^{\sf e} \geq 0$, and $\nu_{m,k}^{\sf e} \geq 0, \forall m \in \CMcal{M}, \forall k \in \CMcal{K}$.
According to the KKT conditions, the optimal solution for the problem in \eqref{eq:Error_obj} should satisfy the followings:
    \begin{align}
        \label{eq:dev_L2_ep}
        &\left.\frac{\partial \cL_2}{\partial \epsilon_{k}}\right|_{\epsilon_k = \epsilon^{\star}_k} =  - \frac{w}{R_{\infty}}\sqrt{\frac{\CMcal{V}_k}{L}}\sqrt{2\pi} \exp\left({\frac{\left(Q^{-1}(\epsilon^{\star}_k)\right)^2}{2}}\right) + \lambda^{\star}_k + \nu^{\star}_k = 0,
        \\
        \label{eq:dev_L2_tau}
        &\left.\frac{\partial \cL_2}{\partial \tau}\right|_{\tau = \tau^{\star}} = \frac{1-w}{\hat{\epsilon}_{\sf max}} - \sum_k \lambda^{\star}_k + \mu^{\star} = 0,
        \\
        \label{eq:dev_L2_del}
        &\left.\frac{\partial \cL_2}{\partial \delta_{m,k}}\right|_{\delta_{m,k} = \delta^{\star}_{m,k}} =
        - \frac{w}{R_{\infty}}\sqrt{\frac{\CMcal{V}_{m,k}^{\sf e}}{L}}\sqrt{2\pi} \exp\left({\frac{\left(Q^{-1}(\delta_{m,k}^{\star})\right)^{2}}{2}}\right)+ \lambda^{\sf e, \star}_{m,k} + \nu^{\sf e, \star}_{m,k} = 0,
        \\
        \label{eq:dev_L2_xi}
        &\left.\frac{\partial \cL_2}{\partial \xi}\right|_{\xi = \xi^{\star}} = \frac{1-w}{\hat{\delta}_{\sf max}} - \sum_{m,k}\lambda^{\sf e, \star}_{m,k} + \mu^{\sf e, \star} = 0,
        \\
        \label{eq:dev_lambda}
        &\lambda^{\star}_k (\tau^{\star} - \epsilon^{\star}_k) = 0,
        \\
        \label{eq:dev_nu}
        &\nu^{\star}_k (\hat{\epsilon}_k - \epsilon^{\star}_k) = 0,
        \\
        \label{eq:dev_mu}
        &\mu^{\star}(\hat{\epsilon}_{\sf max} - \tau^{\star}) = 0,
        \\
        \label{eq:dev_lambda_eve}
        &\lambda^{\sf e, \star}_{m,k} (\xi^{\star} - \delta^{\star}_{m,k}) = 0,
        \\
        \label{eq:dev_nu_eve}
        &\nu^{\sf e, \star}_{m,k} (\hat{\delta}_{m,k} - \delta^{\star}_{m,k}) = 0,
        \\
        \label{eq:dev_mu_eve}
        &\mu^{\sf e, \star}(\hat{\delta}_{\sf max} - \xi^{\star}) = 0,
    \end{align}
    where \eqref{eq:dev_L2_ep} and \eqref{eq:dev_L2_del} come from that $\partial Q^{-1}(x)/\partial x = - \sqrt{2 \pi} \exp\left(\left(Q^{-1}(x)\right)^2/2\right)$.
    Now, we obtain the optimal points $\epsilon^{\star}_k$ and $\delta^{\star}_{m,k}$ by using the derived KKT conditions.
    From \eqref{eq:dev_L2_ep}, \eqref{eq:dev_lambda}, and \eqref{eq:dev_nu}, $\epsilon^{\star}_k$ should be equal to $\tau^{\star}$ or $\hat{\epsilon}_{k}$.
    In the same manner, $\delta^{\star}_{m,k}$ should be equal to $\xi^{\star}$ or $\hat{\delta}_{m,k}$ from \eqref{eq:dev_L2_del}, \eqref{eq:dev_lambda_eve}, and \eqref{eq:dev_nu_eve}.
    Assuming that $\tau^{\star} < \hat{\epsilon}_{\sf max}$ and  $\xi^{\star} < \hat{\delta}_{\sf max}$, we obtain $\mu^{\star} = \mu^{\sf e, \star} = 0$.
    Since we can assume that $\hat{\epsilon}_{\ell-1} < \tau^{\star} < \hat{\epsilon}_{\ell}$ for some $\ell$ and $\hat{\delta}_{j(k)-1,k} < \xi^{\star} < \hat{\delta}_{j(k),k}$ for some $j(k)$, 
    the optimal error probability and information leakage vectors are obtained as
    \begin{align}
        &{\boldsymbol{\epsilon}}^{\star} = [\hat{\epsilon}_1, \cdots, \hat{\epsilon}_{\ell -1},\tau^{\star},\cdots,\tau^{\star}]^{\sf T},
        \\
        &{\boldsymbol{\delta}}_{k}^{\star} = [\hat{\delta}_{1,k}, \cdots, \hat{\delta}_{j(k) -1,k},\xi^{\star},\cdots,\xi^{\star}]^{\sf T}.
    \end{align}
    The information leakage matrix builds upon the derived information leakage vectors as
    \begin{align}
        {\boldsymbol{\Delta}}^{\star} = [{\boldsymbol{\delta}}_{1}^{\star},\cdots,{\boldsymbol{\delta}}_{K}^{\star}].
    \end{align}
    Under the assumption that $\hat{\epsilon}_{\ell-1}<\tau^{\star}<\hat{\epsilon}_{\ell}$ for some $\ell$,  $\nu_{k}^{\star} = 0$ and $\mu^{\star}=0$ for $k\geq \ell$.
    Combining \eqref{eq:dev_L2_ep} and \eqref{eq:dev_L2_tau}, we have
    \begin{align}
        \label{eq:optimal_tau}
        \frac{w}{R_{\infty}}\sum_{k=\ell}^{K}\sqrt{\frac{\CMcal{V}_{k}}{L}}\sqrt{2\pi} \exp\left({\frac{\left(Q^{-1}(\epsilon_{k}^{\star})\right)^{2}}{2}}\right) = \frac{1-w}{\hat{\epsilon}_{\sf max}},  \text{ for } k\geq \ell.
    \end{align}
    We solve \eqref{eq:optimal_tau} with respect to $\epsilon_k^\star$ and then we have $\epsilon_k^\star = \tau^\star$ which is in \eqref{eq:tau_star} for $k \geq \ell$.
    Similarly, we assume that $\hat{\delta}_{j(k)-1,k}<\xi^{\star}<\hat{\delta}_{j(k),k}$, then  $\nu^{\sf e, \star}_{m,k} = 0$ and $\mu^{\sf e, \star}=0$ for $m\geq j(k)$.
    From \eqref{eq:dev_L2_del} and \eqref{eq:dev_L2_xi}, we have
    \begin{align}
        \label{eq:optimal_xi}
        \frac{w}{R_{\infty}}\sum_{m=j(k)}^{M}\sum_{k=1}^{K}\left[\sqrt{\frac{\CMcal{V}_{m,k}^{\sf e}}{L}}\sqrt{2\pi} \exp\left({\frac{\left(Q^{-1}(\delta_{m,k}^{\star})\right)^{2}}{2}}\right)\right] = \frac{1-w}{\hat{\delta}_{\sf max}}, \text{ for } m\geq j(k).
    \end{align}
    Solving \eqref{eq:optimal_xi} with respect to $\delta_{m,k}^\star$, we have $\delta_{m,k}^\star = \xi^\star$ in \eqref{eq:xi_star} for $m\geq j(k)$.
\end{proof}

\bibliographystyle{IEEEtran}
\bibliography{main}

\end{document}